\documentclass[a4paper,10pt,numbers=noenddot,twoside=on,headinclude=true,footinclude=false,headsepline=true]{scrartcl}

\usepackage[utf8]{inputenc}
\usepackage[english]{babel}
\usepackage{color}
\usepackage{amssymb}
\usepackage{amsmath}
\usepackage{amsfonts}
\usepackage{amsopn}
\usepackage{amsbsy}
\usepackage{enumerate}
\usepackage[plainpages=false,pdfpagelabels,breaklinks=true,pdfborderstyle={/S/U/W 1.2}]{hyperref}
\usepackage{graphics}

\numberwithin{equation}{section} 
\numberwithin{table}{section} 
\numberwithin{figure}{section} 

\usepackage[amsmath,thmmarks,thref,hyperref]{ntheorem}

\theoremstyle{plain}
\newtheorem{thm}{Theorem}[section]
\newtheorem{defn}[thm]{Definition}
\newtheorem{lem}[thm]{Lemma}
\newtheorem{cor}[thm]{Corrolary}
\newtheorem{prop}[thm]{Proposition}
\newtheorem{assumption}[thm]{Assumption}

\theorembodyfont{\upshape}
\newtheorem{remark}[thm]{Remark}

\theoremstyle{nonumberplain}

\theoremsymbol{\ensuremath{_\Box}}
\newtheorem{proof}{Proof}

\usepackage[style=alphabetic,
	citestyle=alphabetic,
	firstinits=true,
	backend=biber,
	hyperref=auto,
	sorting=nyt,
	maxcitenames=3,
	maxbibnames=99,
	minnames=3,
	arxiv=pdf,
	url=false,
	isbn=false,
	dateabbrev=true,
	datezeros=true,
	bibencoding=utf8]{biblatex}
\newbibmacro{string+doi}[1]{%
	\iffieldundef{doi}{#1}{\href{http://dx.doi.org/\thefield{doi}}{#1}}}
\DeclareFieldFormat
	{title}{\usebibmacro{string+doi}{\mkbibemph{#1}}\isdot}
\DeclareFieldFormat
	[article,inbook,incollection,inproceedings,patent,thesis,unpublished]
	{title}{\usebibmacro{string+doi}{\mkbibemph{#1}}\isdot}
\DeclareFieldFormat
	[article,inbook,incollection,inproceedings,patent,thesis,unpublished]
	{pages}{#1}
\DeclareFieldFormat
	[article]
	{journaltitle}{#1\isdot}
\DeclareFieldFormat
	[article]
	{volume}{\textbf{#1}}
\DefineBibliographyStrings{english}{pages={}}
\DeclareBibliographyDriver{article}{%
	\usebibmacro{bibindex}%
	\usebibmacro{begentry}%
	\usebibmacro{author/translator+others}%
	\setunit{\labelnamepunct}\newblock
	\usebibmacro{title}%
	\usebibmacro{journal+issuetitle}%
	\setunit*{\addcomma\space}
	\printfield{pages}
	\setunit*{\addcomma\space}
	\printfield{year}
	\usebibmacro{finentry}%
}

\newbibmacro*{journal+issuetitle}{%
	\usebibmacro{journal}%
	\setunit*{\addspace}%
	\iffieldundef{series}
	{}
	{\newunit
		\printfield{series}%
		\setunit{\addspace}}%
	\iffieldundef{volume}
		{}
		{\printfield{volume}%
		\setunit{\addspace}}%
	\setunit*{\addcolon\space}%
	\usebibmacro{issue}%
	\newunit}

\usepackage[scaled]{beramono}
\usepackage{helvet}
\usepackage[charter]{mathdesign} 
\SetMathAlphabet{\mathcal}{normal}{OMS}{cmsy}{m}{n} 
\SetMathAlphabet{\mathcal}{bold}{OMS}{cmsy}{m}{n} 
\providecommand{\ie}{i.~e.~}
\providecommand{\eg}{e.~g.~}
\providecommand{\cf}{cf.~}

\providecommand{\R}{\mathbb{R}}

\providecommand{\C}{\mathbb{C}}
\renewcommand{\C}{\mathbb{C}}

\providecommand{\T}{\mathbb{T}}
\renewcommand{\T}{\mathbb{T}}
\providecommand{\N}{\mathbb{N}}
\providecommand{\Z}{\mathbb{Z}}
\providecommand{\ii}{\mathrm{i}}
\providecommand{\e}{\mathrm{e}}

\providecommand{\Hil}{\mathcal{H}}

\providecommand{\eps}{\varepsilon}
\providecommand{\Cont}{\mathcal{C}}

\providecommand{\ker}{\mathrm{ker} \, }
\providecommand{\ran}{\mathrm{ran} \, }

\providecommand{\ker}{\mathrm{ker} \,}

\providecommand{\dd}{\mathrm{d}}
\providecommand{\id}{\mathrm{id}}

\providecommand{\order}{\mathcal{O}}
\providecommand{\Fourier}{\mathcal{F}}

\providecommand{\abs}[1]{\left \lvert #1 \right \rvert}
\providecommand{\sabs}[1]{\lvert #1 \vert}
\providecommand{\babs}[1]{\bigl \lvert #1 \bigr \rvert}

\providecommand{\norm}[1]{\left \lVert #1 \right \rVert}
\providecommand{\snorm}[1]{\lVert #1 \rVert}
\providecommand{\bnorm}[1]{\bigl \lVert #1 \bigr \rVert}
\providecommand{\Bnorm}[1]{\Bigl \lVert #1 \Bigr \rVert}
\providecommand{\scpro}[2]{\left \langle #1 , #2 \right \rangle}

\providecommand{\bscpro}[2]{\bigl \langle #1 , #2 \bigr \rangle}

\providecommand{\sopro}[2]{\vert #1 \rangle \langle #2 \vert}

\providecommand{\s}[1]{\mathcal{#1}}

\providecommand{\ncint}{\mathrel{{\ooalign{$\int$\cr\kern+.07em\raise.15ex\hbox{$\pmb{\scriptstyle-}$}\cr}}}}           \providecommand{\ncpartial}{\mathrel{{\ooalign{$\partial$\cr\kern+.29em\raise.79ex\hbox{$\pmb{\scriptstyle-}$}\cr}}}}

\providecommand{\e}{\mathrm{e}}
\providecommand{\ii}{\mathrm{i}}

\providecommand{\Hil}{\mathfrak{H}}
\providecommand{\Cont}{\mathcal{C}}
\providecommand{\Schwartz}{\mathcal{S}}
\providecommand{\Zak}{\mathcal{Z}}

\providecommand{\Msymb}{\mathcal{M}}
\providecommand{\Op}{\mathfrak{Op}}

\providecommand{\directsum}{\oplus}
\providecommand{\orthsum}{\oplus_{\perp}}

\providecommand{\Maxwell}{\mathbf{M}}

\providecommand{\Hper}{\Hil_0}
\providecommand{\Jper}{J_0}
\providecommand{\Gper}{G_0}
\providecommand{\Pper}{P_0}
\providecommand{\Qper}{Q_0}
\providecommand{\Mper}{\Maxwell_0}

\providecommand{\HperT}{\mathfrak{h}_0}

\providecommand{\BZ}{\mathbb{B}}
\providecommand{\WS}{\mathbb{W}}

\providecommand{\Jrot}{J_{\Rot}}

\providecommand{\Jreg}{\Jper^{\mathrm{reg}}}
\providecommand{\Greg}{\Gper^{\mathrm{reg}}}
\providecommand{\Qreg}{\Qper^{\mathrm{reg}}}
\providecommand{\Preg}{\Pper^{\mathrm{reg}}}

\providecommand{\Jphys}{\mathbf{J}}
\providecommand{\Gphys}{\mathbf{G}}
\providecommand{\Mphys}{\mathbf{M}}
\providecommand{\Pphys}{\mathbf{P}}
\providecommand{\Qphys}{\mathbf{Q}}

\providecommand{\Hoer}[1]{S^{#1}_{\rho}}
\providecommand{\Hoerm}[2]{S^{#1}_{#2}}
\providecommand{\Hoereq}[1]{S^{#1}_{\rho,\mathrm{eq}}}
\providecommand{\Hoermeq}[2]{S^{#1}_{#2,\mathrm{eq}}}

\providecommand{\SemiHoereq}[1]{A S^{#1}_{\rho,\mathrm{eq}}}
\providecommand{\SemiHoermeq}[2]{A S^{#1}_{#2,\mathrm{eq}}}

\providecommand{\eff}{\mathrm{eff}}

\providecommand{\eff}{\mathrm{eff}}

\providecommand{\Div}{\mathbf{Div}}
\renewcommand{\div}{\mathbf{div}}
\providecommand{\Rot}{\mathbf{Rot}}
\providecommand{\curl}{\nabla_x^{\times}}
\providecommand{\Grad}{\mathbf{Grad}}
\providecommand{\div}{\mathbf{div}}

\providecommand{\grad}{\nabla_x}

\providecommand{\Index}{\mathcal{I}}

\providecommand{\derln}{\Upsilon}

\providecommand{\Schwartz}{\mathcal{S}}
\providecommand{\domain}{\mathfrak{D}}
\providecommand{\domainT}{\mathfrak{d}}

\bibliography{./bibliography}

\usepackage{units}
\usepackage{diagxy}

\title{The Perturbed Maxwell Operator as Pseudodifferential Operator}
\author{Giuseppe De Nittis${}^{\ast}$ \& Max Lein${}^{\star}$}

\begin{document}

\maketitle
\vspace{-9mm}
\begin{center}
	$^{\ast}$ Department Mathematik, Universität Erlangen-Nürnberg \linebreak
	Cauerstrasse 11, D-91058 Erlangen, Germany \linebreak
	{\footnotesize \href{mailto:denittis@math.fau.de}{\texttt{denittis@math.fau.de}}}
	\medskip
	\\
	$^{\star}$ Kyushu University, Faculty of Mathematics \linebreak
	744 Motooka, Nishiku, Fukuoka, 819-0395, Japan \linebreak
	{\footnotesize \href{mailto:max.lein@me.com}{\texttt{max.lein@me.com}}}
\end{center}
\begin{abstract}
	As a first step to deriving effective dynamics and ray optics, we prove that the perturbed periodic Maxwell operator in $d = 3$ can be seen as a pseudo\-differential operator. This necessitates a better understanding of the periodic Maxwell operator $\Mper$. In particular, we characterize the behavior of $\Mper$ and the physical initial states at small crystal momenta $k$ and small frequencies. Among other things, we prove that generically the band spectrum is symmetric with respect to inversions at $k = 0$ and that there are exactly $4$ ground state bands with approximately linear dispersion near $k = 0$. 
\end{abstract}

\newpage
\tableofcontents

\pagestyle{headings}

\section{Introduction} 
\label{intro}
Photonic crystals are to the transport of light (electromagnetic waves) what crystalline solids are to the transport of electrons \cite{Joannopoulos_Johnson_Winn_Meade:photonic_crystals:2008}. Progress in the manufacturing techniques have allowed physicists to engineer photonic crystals with specific properties -- which in turn has stimulated even more theoretical studies. One topic which has seen relatively little attention, though, is the derivation of \emph{effective dynamics} in perturbed photonic crystals for states from a narrow range of intermediate frequencies (\eg \cite{Onoda_Murakami_Nagaosa:geometrics_optical_wave-packets:2006,Raghu_Haldane:quantum_Hall_effect_photonic_crystals:2008,Allaire_Palombaro_Rauch:diffractive_Bloch_wave_packets_Maxwell:2012,Esposito_Gerace:photonic_crystals_broken_TR_symmetry:2013}). Mathematically rigorous results are even more scarce: apart from \cite{Markowich_Poupaud:Maxwell_homogenization_energy_density:1996} concerning only the unperturbed case, the only rigorous work covering \emph{second}-order perturbations is by Allaire, Palombaro and Rauch \cite{Allaire_Palombaro_Rauch:diffractive_Bloch_wave_packets_Maxwell:2012}. Hence, the correct form of the subleading-order terms has not yet been established -- rigorously or non-rigorously.

This paucity of results motivated the two authors to apply a perturbation scheme developed by Panati, Spohn and Teufel \cite{PST:sapt:2002,PST:effective_dynamics_Bloch:2003}, \emph{space-adiabatic perturbation theory}, to derive effective dynamics and ray optics equations for adiabatically perturbed Maxwell operators. Among other things, we settle the important question about the correct form of the next-to-leading order terms in the ray optics equations; these terms are necessary to explain topological effects in photonic crystals. The current paper is a preliminary, but necessary step to implement space-adiabatic perturbation theory \cite{DeNittis_Lein:sapt_photonic_crystals:2013}: we establish that the Maxwell operator can be seen as a \emph{semiclassical pseudodifferential operator} ($\Psi$DO) with band structure defined over the cotangent bundle over the Brillouin torus.

This is not just the content of an innocent lemma, it turns out there are quite a few technical and conceptual hurdles to overcome. To mention but one, we need a better understanding of the band structure of the periodic Maxwell operator. Despite the body of work on periodic Maxwell operators (see \eg \cite{Kuchment:math_photonic_crystals:2001} for a review), proofs of rather fundamental results are either scattered throughout the literature or, in some cases, seem to have not been published at all. 
\medskip

\noindent
Before we expound on this point in more detail, let us recall the $L^2$-theory of electromagnetism first established in \cite{Birman_Solomyak:L2_theory_Maxwell_operator:1987}. The two dynamical equations 
\begin{align}
	\partial_t \mathbf{E} &= + \eps^{-1} \nabla_x \times \mathbf{H} 
	, 
	&&
	\partial_t \mathbf{H} = - \mu^{-1} \nabla_x \times \mathbf{E}
	, 
	\label{intro:eqn:traditional_Maxwell_equations_dynamics} 
\end{align}
can be recast as a time-dependent Schrödinger equation
\begin{align}
	\ii \partial_t \Psi = \Mphys_w \Psi 
	\label{intro:eqn:Maxwell_Schroedinger_type}
\end{align}
where $\Psi = (\mathbf{E},\mathbf{H})$ consists of the electric field $\mathbf{E} = (E_1,E_2,E_3)$ and the magnetic field $\mathbf{H} = (H_1 , H_2 , H_3)$, and 
\begin{align}
	\Mphys_w &:= \left (
	\begin{matrix}
		0 & + \ii \, \eps^{-1} \, \nabla_x^{\times} \\
		- \ii \, \mu^{-1} \, \nabla_x^{\times} & 0 \\
	\end{matrix}
	\right )
	\label{intro:eqn:Maxwell_operator}
\end{align}
is the \emph{Maxwell operator}. Here we used $\nabla_x^{\times}$ as shorthand for the curl (\cf Appendix~\ref{appendix:curl}). The second set of Maxwell equation which imposes the absence of sources, 
\begin{align}
	\nabla_x \cdot \eps \mathbf{E} &= 0 
	, 
	&&
	\nabla_x \cdot \mu \mathbf{H} = 0 
	, 
	\label{intro:eqn:traditional_Maxwell_equations_no_source}
\end{align}
enter as a constraint on the initial conditions for equation~\eqref{intro:eqn:Maxwell_Schroedinger_type} or, equivalently, one can restrict the domain to the physical states of $\Mphys_w$ (see Section~\ref{Maxwell:generic}). We shall always make the following assumptions on the material weights $w = (\eps,\mu)$: 
\begin{assumption}[Material weights]\label{intro:assumption:eps_mu_generic}
	Assume $\eps , \mu \in L^{\infty} \bigl ( \R^3 , \mathrm{Mat}_{\C}(3) \bigr )$ are hermitian-matrix-valued functions which are bounded away from $0$ and $+\infty$, \ie $0 < c \, \id_{\R^3} \leq \eps , \mu \leq C \, \id_{\R^3}$ for some $0 < c \leq C < \infty$. We say the material weights $(\eps,\mu)$ are \emph{real} iff their entries are all real-valued functions. 
\end{assumption}
These assumptions are rather natural in the setting we are interested in: First of all, asking for boundedness of $\eps$ and $\mu$ only instead of continuity is necessary to include the most common cases, because many photonic crystals are made by alternating two different materials, \eg a dielectric and air, in a periodic fashion. The selfadjointness of the multiplication operator defined by the \emph{electric permittivity tensor} $\eps^* = \eps$ and the \emph{magnetic permeability tensor} $\mu^* = \mu$ ensure that the medium neither absorbs nor amplifies electromagnetic waves. The positivity of $\eps$ and $\mu$ excludes the case of metamaterials with negative refraction indices (see \eg \cite{Smith_Padilla_et_al:metamaterials:2000}); moreover, combined with the boundedness away from $0$ and $+\infty$, it implies that $\eps^{-1}$ and $\mu^{-1}$ exist as bounded operators which again satisfy Assumption~\ref{intro:assumption:eps_mu_generic}. Lastly, our assumptions also include the interesting case of \emph{gyrotropic photonic crystals} where the offdiagonal entries of $\eps = \eps^*$ and $\mu = \mu^*$ are complex-valued functions.

Under these assumptions, we can proceed with a rigorous definition of the Maxwell operator \eqref{intro:eqn:Maxwell_operator}: it can be conveniently factored into
\begin{align}
	\Mphys_w = W \, \Rot 
	\, .
	\label{Maxwell:eqn:Maxwell_physical} 
\end{align}
where the first term is the bounded operator involving the weights 
\begin{align}
	W(\hat{x}) 
	:= \left (
	\begin{matrix}
		\eps^{-1}(\hat{x}) & 0 \\ 
		0 & \mu^{-1}(\hat{x}) \\
	\end{matrix}
	\right )
	\label{Maxwell:eqn:max1}
\end{align}
and the free Maxwell operator
\begin{align}\label{intro:eqn:Rot}
	\Rot := \left (
	\begin{matrix}
		0 & + \ii \, \curl \\
		- \ii \, \curl & 0 \\
	\end{matrix}
	\right ) 
	= \left (
		\begin{matrix}
			0 & + \ii \, \mathbf{curl} \\
			- \ii \, \mathbf{curl} & 0 \\
		\end{matrix}
		\right )
	\, .
\end{align}
$\Rot$ equipped with the domain $\domain := \domain(\Rot) \subset L^2(\R^3,\C^6)$ is selfadjoint (see Appendix~\ref{appendix:curl} for a precise characterization of $\domain$). For reasons that will be clear in the following, we refer to \eqref{Maxwell:eqn:Maxwell_physical} as the \emph{physical representation} of the Maxwell operator. From the representation \eqref{Maxwell:eqn:Maxwell_physical} one gets two immediate consequences: first, $\domain(\Maxwell_w) = \domain$ since $W$ is bounded and second, $\Maxwell_w$ is \emph{not} self-adjoint on $L^2(\R^3,\C^6)$. In order to cure the lack of selfadjointness one introduces the \emph{weighted} scalar product
\begin{align}
	\bscpro{\Psi}{\Phi}_w := \bscpro{\Psi}{W^{-1} \Phi}_{L^2(\R^3,\C^6)} 
	= \bscpro{W^{-1} \Psi}{\Phi}_{L^2(\R^3,\C^6)} 
	\, . 
	\label{intro:eqn:relation_weighted_unweighted_scalar_product}
\end{align}
on the Banach space $L^2(\R^3,\C^6)$, and we will denote this Hilbert space with $\Hil_w$. Then, one can show that the Maxwell operator $\Mphys_w$ is self-adjoint on $\domain \subset \Hil_w$ (\cf Theorem~\ref{Maxwell:thm:selfadjointness}). Only with respect to the correctly weighted scalar product, the evolutionary semigroup $\e^{- \ii t \Mphys_w}$ is unitary -- which physically corresponds to conservation of field energy $\mathcal{E} \bigl ( \mathbf{E}(t),\mathbf{H}(t) \bigr ) = \mathcal{E}(\mathbf{E},\mathbf{H})$,
\begin{align*}
	\mathcal{E}(\mathbf{E},\mathbf{H}) &= \frac{1}{2} \int_{\R^3} \dd x \, \mathbf{E}(x) \cdot \eps(x) \mathbf{E}(x) + \frac{1}{2} \int_{\R^3} \dd x \, \mathbf{H}(x) \cdot \mu(x) \mathbf{H}(x)
	\\
	&
	= \frac{1}{2} \, \bnorm{(\mathbf{E},\mathbf{H})}_w^2 
	\, . 
\end{align*}
\emph{Periodic Maxwell operators} describe photonic crystals; here, the material weights $\eps$ and $\mu$ are periodic with respect to some lattice $\Gamma$. As the analog of periodic Schrödinger operators, one can use Bloch-Floquet theory to analyze the properties of $\Mphys_w$ (\cf Section~\ref{periodic}). Hence, many properties of photonic crystals mimic those of crystalline solids (both physically and mathematically). However, the rapidly increasing interest for  photonic crystals resides in the fact that, as they are artificially created by patterning several materials, they can be engineered to have certain desired properties. To name one example, one of the early successes was to design a \emph{photonic semiconductor} with a band gap in the frequency spectrum \cite{Johnson_Joannopoulos:3d_photonic_crystal_band_gap:2000,Joannopoulos_Johnson_Winn_Meade:photonic_crystals:2008}. Such a “semiconductor for light” is of great interest to the quantum optics community (\eg \cite{Yablonovitch:photonic_band_gap:1993}). 

Since perfectly periodic media are only a mathematical abstraction, one is led to study more realistic models of photonic crystals. One well-explored possibility is to include effects of disorder by interpreting $\eps$ and $\mu$ as random variables and leads to the “Anderson localization of light” (see \eg \cite{John:Anderson_localization_light:1991,Figotin_Klein:localization_classical_waves_I:1996,Figotin_Klein:localization_classical_waves_II:1997} and references therein). We will concern ourselves with another class of perturbations where the perfectly periodic weights $\eps$ and $\mu$ are modulated slowly, 
\begin{align}
	\eps_{\lambda}(x) := \frac{\eps(x)}{\tau_{\eps}(\lambda x)^2} 
	\, , 
	&&
	\mu_{\lambda}(x) := \frac{\mu(x)}{\tau_{\mu}(\lambda x)^2} 
	\, . 
	\label{intro:eqn:slow_modulation_material_constants}
\end{align}
The perturbation parameter $\lambda \ll 1$ quantifies the separation of spatial scales on which $(\eps,\mu)$ and the scalar \emph{modulation functions} $(\tau_{\eps},\tau_{\mu})$ vary. The latter are assumed to verify the following 
\begin{assumption}[Modulation functions]\label{Maxwell:assumption:modulation_functions}
	Suppose $\tau_{\eps} , \tau_{\mu} \in \Cont^{\infty}_{\mathrm{b}}(\R^3)$ are bounded away from $0$ and $+\infty$ as well as $\tau_{\eps}(0) = 1$ and $\tau_{\mu}(0) = 1$. 
\end{assumption}
To shorten the notation, we define $\Mphys_{\lambda} := \Mphys_{(\eps_{\lambda},\mu_{\lambda})}$ and $\Hil_{\lambda} := \Hil_{(\eps_{\lambda},\mu_{\lambda})}$. 
\medskip

\noindent
As mentioned in the very beginning our goal is to rigorously derive both, the effective “quantum-like” and “semiclassical” dynamics for perturbed Maxwell operators $\Mphys_{\lambda}$ in the adiabatic limit $\lambda \ll 1$ \cite{DeNittis_Lein:sapt_photonic_crystals:2013}. Apart from ray optics, we will derive \emph{effective light dynamics} $\e^{- \ii t \Mphys_{\eff}}$ which approximate the full light dynamics $\e^{- \ii t \Mphys_{\lambda}}$ for initial states supported in a narrow range of frequencies, 
\begin{align}
	\Bnorm{\bigl ( \e^{- \ii t \Mphys_{\lambda}} - \e^{- \ii t \Mphys_{\eff}} \bigr ) \, \pmb{\Pi}_{\lambda}}_{\Hil_{\lambda}} &= \order(\lambda^{\infty}) 
	\, . 
	\label{intro:eqn:approximate_dynamics}
\end{align}
$\pmb{\Pi}_{\lambda}$ is the projection on the superadiabatic subspace associated with a narrow range of frequencies and, up to a unitary transformation, the effective operator $\Mphys_{\eff}$ can be constructed order-by-order in $\lambda$ as the Weyl quantization $\Op_{\lambda}(\Msymb_{\eff})$ of a semiclassical symbol; in case additional assumptions are placed on the frequency bands, the leading-order terms are given by 
\begin{align*}
	\Msymb_{\eff}(r,k) &= \sum_{n \in \Index} \tau_{\eps}(r) \, \tau_{\mu}(r) \, \omega_n(k) \, \sopro{\chi_n}{\chi_n} 
	+ \order(\lambda) 
	\, . 
\end{align*}
Here, the $\omega_n$ are the Bloch frequency band functions and $\chi_n$ denotes a fixed orthonormal basis in the reference space \cite[Theorem~3.1]{DeNittis_Lein:sapt_photonic_crystals:2013}. As usual one can also prove that the subleading-order terms of $\Msymb_{\eff}(r,k)$ contain geometric quantities such as the Berry connection. 

Similarly, the superadiabatic projection $\pmb{\Pi}_{\lambda}$ is also constructed on the level of symbols in terms of $\pmb{\Msymb}_{\lambda}$, the symbol of the Maxwell operator, and hence, proving that the Maxwell operator is a $\Psi$DO associated to a semiclassical symbol is the first order of business. 
\begin{thm}\label{intro:thm:perturbed_Maxwell_psuedo}
	Suppose Assumptions~\ref{periodic:assumption:periodic_eps_mu} on the material weights $(\eps,\mu)$ and \ref{Maxwell:assumption:modulation_functions} on the modulation functions $(\tau_{\eps},\tau_{\mu})$ are satisfied. Then the Maxwell operator (in Zak representation) $\Mphys^{\Zak}_{\lambda} = \Op_{\lambda}(\pmb{\mathcal{M}}_{\lambda})$ is the pseudodifferential operator associated to 
	\begin{align}
		\pmb{\mathcal{M}}_{\lambda}(r,k) &= 
		\left (
		\begin{matrix}
			\tau_{\eps}^2(r) & 0 \\
			0 & \tau_{\mu}^2(r) \\
		\end{matrix}
		\right )
		\, \Mper(k)
		+ \notag \\
		&\; \quad 
		+ \lambda \, W \, 
		\left (
		\begin{matrix}
			0 & - \ii \, \tau_{\eps}(r) \, \bigl ( \nabla_r \tau_{\eps} \bigr )^{\times}(r) \\
			+ \ii \, \tau_{\mu}(r) \, \bigl ( \nabla_r \tau_{\mu} \bigr )^{\times}(r) & 0 \\
		\end{matrix}
		\right )
		\label{intro:eqn:symbol_perturbed_Maxwell}
	\end{align}
	where 
	\begin{align*}
		\Mper(k) :& \negmedspace= W \, \Rot(k) 
		\\
		:& \negmedspace= \left (
		\begin{matrix}
			\eps^{-1}(\hat{y}) & 0 \\
			0 & \mu^{-1}(\hat{y}) \\
		\end{matrix}
		\right ) \, \left (
		\begin{matrix}
			0 & - (- \ii \nabla_y + k)^{\times} \\
			+ (- \ii \nabla_y + k)^{\times} & 0 \\
		\end{matrix}
		\right ) 
	\end{align*}
	is the periodic Maxwell operator acting on the fiber at $k$ defined in terms of the weight operator $W$ and the free Maxwell operator $\Rot(k)$. 
	The function $\pmb{\mathcal{M}}_{\lambda} \in \SemiHoermeq{1}{1} \left ( \mathcal{B} \bigl ( \domainT , L^2(\T^3,\C^6) \bigr ) \right )$ is an equivariant semiclassical operator-valued symbol in the sense of Definition~\ref{pseudo_Maxwell:defn:equivariant_symbols}.
\end{thm}
For the precise definitions and the proof, we refer to Section~\ref{pseudo_Maxwell}. 
\medskip

\noindent
Despite the similarities to the case of the Bloch electron \cite{PST:effective_dynamics_Bloch:2003}, applying space-adiabatic perturbation theory to photonic crystals required us to solve numerous technical and conceptual problems. In addition to defining pseudo\-differential operators on weighted $L^2$-spaces, one other major difficulty is to make $\order(\lambda^n)$ estimates in norm, because the norm \emph{also} depends on $\lambda$ (see \eg equation~\eqref{intro:eqn:approximate_dynamics}). Such estimates are crucial when one wants to make sense of perturbation expansions of operators. This conceptual problem is solved by introducing a $\lambda$-independent auxiliary representation (\cf Section~\ref{Maxwell:perturbed}).

However, the biggest obstacle to control the symbol $\pmb{\Msymb}_{\lambda}$ is to gain a better understanding of the \emph{periodic} Maxwell operator $\Mper(k)$ and its band structure. In particular, pseudodifferential theory requires us to understand the \emph{pointwise} behavior of $\Mper(k)$ and associated objects. Even though $k \mapsto \Mper(k)$ is linear and defined on a $k$-independent domain, and thus trivially analytic, the splitting of the \emph{fiber} Hilbert space $\HperT = \Jper(k) \oplus_{\perp} \Gper(k)$ into physical and unphysical states is not even \emph{continuous} at $k = 0$. Here, $\HperT$ is defined as the Banach space $L^2(\T^3,\C^6)$ equipped with a scalar product analogous to \eqref{intro:eqn:relation_weighted_unweighted_scalar_product}, and elements of $\Jper(k)$ satisfy the source-free condition on the fiber space. We characterize how this discontinuity enters into the band structure of $\Mper(k)$, and show that it is connected to the \emph{ground state bands}, \ie those frequency bands which go to $0$ linearly as $k \rightarrow 0$. The precise band structure of $\Mper^{\Zak} = \int_{\BZ}^{\oplus} \dd k \, \Mper(k)$ is studied in great detail in Section \ref{periodic:fiber_operators} where the following result is proven:
\begin{thm}[The band picture of $\Mper^{\Zak}$]\label{intro:thm:band_picture}
	Suppose $\eps$ and $\mu$ satisfy Assumption~\ref{periodic:assumption:periodic_eps_mu}. 
	\begin{enumerate}[(i)]
		\item For each $n \in \Z$, the band functions $\R^3 \ni k \mapsto \omega_{n}(k)$ are continuous, analytic away from band crossings and $\Gamma^*$-periodic. 
		\item If the weights $(\eps,\mu)$ are real, then for all $n \in \Z$, there exists $j \in \Z$ such that $\omega_{n}(k) = - \omega_{j}(-k)$ holds for all $k \in \R^3$. 
		\item $\Mper^{\Zak}$ has $4$ ground state bands indexed by the set $\Index_{\mathrm{gs}}$ which are characterized as follows: 
		\begin{enumerate}[(1)]
			\item $\omega_n(k) = 0 $ $\Leftrightarrow$ $n \in \Index_{\mathrm{gs}}$ and $k = 0$. 
			\item $\displaystyle \omega_n(k) = \pm c_n(\underline{k}) \abs{k} + o(\abs{k})$ holds for $n \in \Index_{\mathrm{gs}}$ where the $c_n(\underline{k})$ are the positive eigenvalues of the matrix~\eqref{periodic:eqn:k_cdot_A} for the unit vector $\underline{k} := \frac{k}{\sabs{k}}$. 
		\end{enumerate}
	\end{enumerate}
\end{thm}
The content of Theorem \ref{intro:thm:band_picture} is sketched in Figure~\ref{periodic:fig:frequency_bands}. Among other things, we prove that the ground state bands of the Maxwell operator always have a doubly degenerate conical intersection at $k = 0$ and $\omega = 0$. 
\begin{figure}[ht]
	\centering
		\boxed{\resizebox{110mm}{!}{\includegraphics{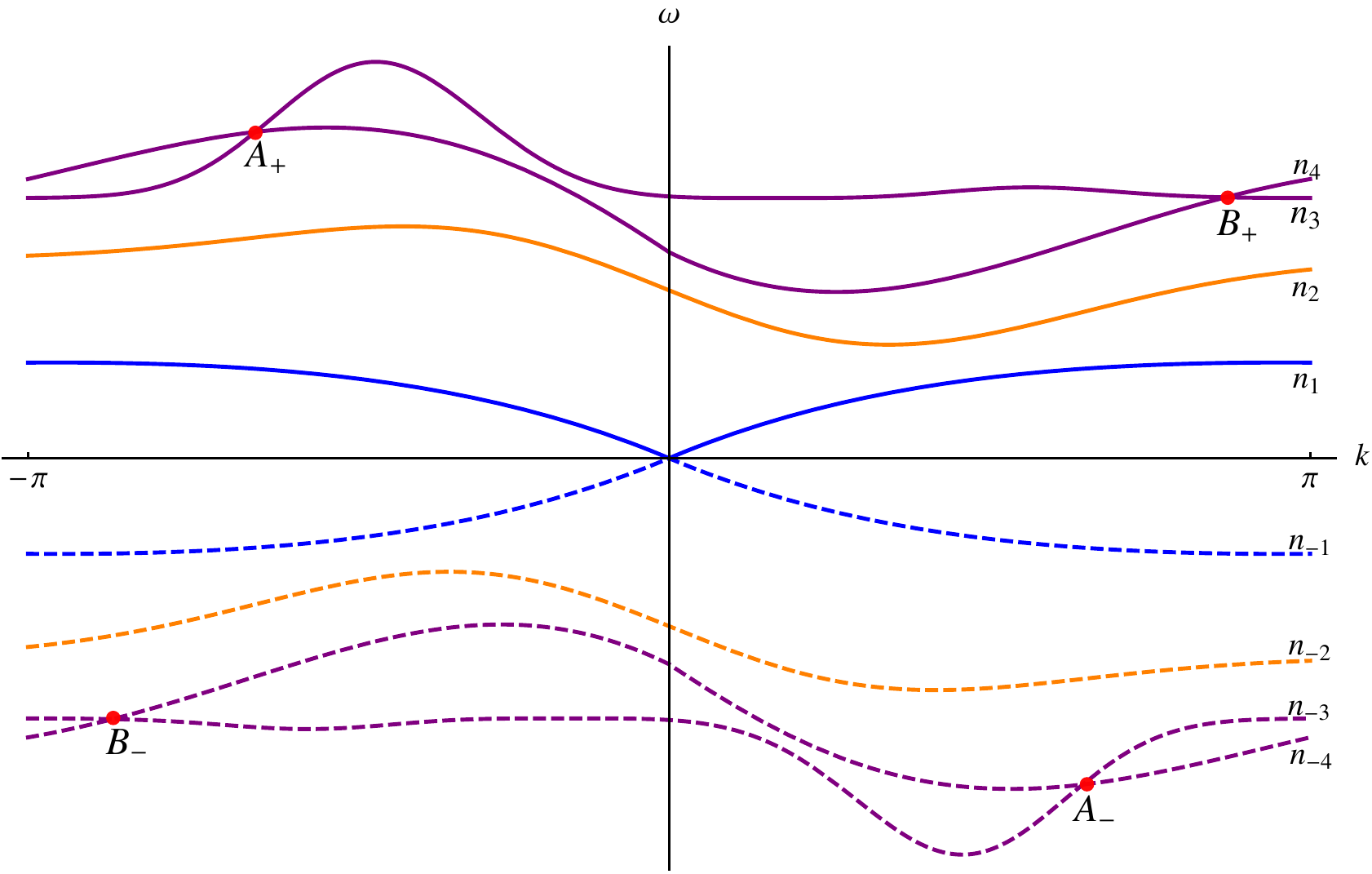}}}
	\caption{A sketch of a typical band spectrum of $\Mper(k) \vert_{\Jper(k)}$. The $2+2$ ground state bands with linear dispersion around $k = 0$ are blue. Positive frequency bands are drawn using solid lines while the lines for the symmetrically-related negative frequency bands are in the same color, but dashed. }
	\label{periodic:fig:frequency_bands}
\end{figure}
\medskip

\noindent
The remainder of the paper is dedicated to explaining and proving Theorem~\ref{intro:thm:perturbed_Maxwell_psuedo} and Theorem~\ref{intro:thm:band_picture}: In Section~\ref{Maxwell}, we give some basic facts on the Maxwell operator. Section~\ref{periodic} is devoted to the study of the properties of the periodic operator $\Mper^{\Zak}$ with a particular attention to the analysis of the band picture. Finally, in Section~\ref{pseudo_Maxwell} where discuss pseudodifferential theory on weighted Hilbert spaces and finish the proof of Theorem~\ref{intro:thm:perturbed_Maxwell_psuedo}. For the benefit of the reader, we have included some auxiliary results in Appendix~\ref{appendix:curl}. 
\medskip

\noindent
Before we proceed, let us collect some conventions and introduce notation used throughout the remainder of the paper.

\subsection{Notation and remarks} 
\label{intro:notation}
The Maxwell operator is naturally defined on \emph{weighted} $L^2$-spaces $\Hil_w$ where the scalar product is weighted by the tensors $w = (\eps,\mu)$ according to the prescription \eqref{intro:eqn:relation_weighted_unweighted_scalar_product}. We will use capital greek letters such as $\Psi$ and $\Phi$ to denote elements of $\Hil_w$ and small greek letters with the appropriate index to indicate they are the electric (first three) or the magnetic (last three) component\footnote{Note that even though physical electromagnetic fields are real-valued, we assume $\Psi \in \Hil_w$ takes values in the complex vector space $\C^6$, and hence our distinction in notation to the physical fields $(\mathbf{E},\mathbf{H})$. It turns out to be crucial in the analysis of photonic crystals to admit complex solutions. \label{intro:footnote:complex_fields}} , for instance $\Psi=(\psi^E,\psi^H)$ and $\Phi=(\phi^E,\phi^H)$. Componentwise the scalar product \eqref{intro:eqn:relation_weighted_unweighted_scalar_product} reads
\begin{align}
	\bscpro{\Psi}{\Phi}_w :\negmedspace & = \int_{\R^3} \dd x \,  \, \psi^E(x) \cdot \eps(x) \phi^E(x) +  \int_{\R^3} \dd x \, \psi^H(x) \cdot \mu(x) \phi^H(x) 
	\, . 
	\label{intro:eqn:scalar_product_explicit}
\end{align}
Let us point out that with this convention the complex conjugation is implicit in the scalar product like  $a \cdot b := \sum_{j = 1}^N \overline{a_j} \, b_j$ on $\C^N$. Equation \eqref{intro:eqn:scalar_product_explicit} leads to the natural (orthogonal) splitting 
\begin{align}
	\Hil_w :\negmedspace &= L^2_{\eps}(\R^3,\C^3) \orthsum L^2_{\mu}(\R^3,\C^3) 
	\, , 
	\notag 
\end{align}
where $L^2_{\eps}(\R^3,\C^3)$ is the Banach space $L^2(\R^3,\C^3)$ with the scalar product twisted by the tensor $\eps$ and similarly for $\mu$.

Even though the Hilbert space structure of $\Hil_w$ depends crucially on the weights $w = (\eps,\mu)$, the Assumption \ref{intro:assumption:eps_mu_generic} implies the equivalence of the norm $\norm{\cdot}_w$ with the usual $L^2(\R^3,\C^6)$-norm $\norm{\cdot}$. This means that $\Hil_w$ agrees with the usual $L^2(\R^3,\C^6)$ as Banach spaces. For many arguments in this paper, only the Banach space structure of $\Hil_w$ is important, and thus, whenever convenient, we will use the canonical identification of $\Hil_w \simeq L^2(\R^3,\C^6)$. In particular, any closed operator $\bf{T}$ on $\Hil_w$ can also be seen as a closed operator on $L^2(\R^3,\C^6)$ which we denote with the same symbol. We will use the same notation for weighted $L^2$-spaces over $\T^3$: for instance, the Hilbert space 
\begin{align*}
	\HperT := L^2_{\eps}(\T^3,\C^3) \orthsum L^2_{\mu}(\T^3,\C^3)
\end{align*}
is defined as the Banach space $L^2(\T^3,\C^6)$ equipped with a scalar product analogous to equation~\eqref{intro:eqn:scalar_product_explicit}. 

Let us turn to conventions regarding operators: Suppose $A : \domain_0(A) \subseteq \mathfrak{B}_1 \longrightarrow \mathfrak{B}_2$ is a possibly unbounded linear operator between the Banach spaces $\mathfrak{B}_1$ and $\mathfrak{B}_2$ defined on the dense domain $\domain_0(A)$. The operator $A$ is called \emph{closable} if and only if for every $\{\psi_n\} \subset \domain_0(A)$ such that $\psi_n \to 0$, then also $A \psi_n \to 0$. The \emph{closure} of the operator $A$ (still denoted with the same symbol) is the extension of $A$ to $\domain(A) := \overline{\domain_0(A)}^{\norm{\cdot}_A}$ with respect to the \emph{graph norm} 
\begin{align}
	\norm{\psi}_A := \sqrt{\snorm{\psi}_{\mathfrak{B}_1}^2 + \snorm{A \psi}_{\mathfrak{B}_2}^2} 
	. 
	\label{intro:eqn:definition:graph_norm}
\end{align}
When $\domain_0(A) = \domain(A)$, the operator $A$ is said to be \emph{closed}. A \emph{core} $\mathfrak{C}$ of a closed operator is any subset of $\domain(A)$ which is dense with respect to $\norm{\cdot}_A$. Given any closed operator $A : \mathcal{B}_1 \longrightarrow \mathcal{B}_2$ between Banach spaces, the kernel (or null space) and range of $A$ are defined as 
\begin{align*}
	\ker A :& \negmedspace= \bigl \{ \psi \in \s{B}_1 \; \vert \; A \psi = 0 \bigr \} \subset \domain(A) 
	\subseteq \mathfrak{B}_1 
	, 
	\\
	\mathrm{ran}_0 \, A :& \negmedspace= \bigl \{ A \psi \; \; \vert \; \; \psi \in \domain(A) \bigr \} 
	\subseteq \mathfrak{B}_2 
\end{align*}
While $\ker A$ is automatically a closed subspace of $\mathfrak{B}_1$, in general $\mathrm{ran}_0 \, A$ is not. For this reason, we need to introduce its closure $\ran A := \overline{\mathrm{ran}_0 \, A}^{\norm{\cdot}_{\mathfrak{B}_2}}$.

Other properties, most notably selfadjointness, crucially depend on the scalar product. Whenever the Hilbert structure of $\Hil_w$ is important, we will make this explicit either in the text or in notation. To give one example, we distinguish between the \emph{direct} sum $J \directsum G$ and the \emph{orthogonal} sum $J \orthsum G$ of vector spaces. 

We found it convenient to use the shorthand $v^{\times} \psi := v \times \psi$ to associate the antisymmetric matrix 
\begin{align}
	v^{\times} = \left (
	\begin{matrix}
		0 & -v_3 & +v_2 \\
		+v_3 & 0 & -v_1 \\
		-v_2 & +v_1 & 0 \\
	\end{matrix}
	\right )
	\, 
	\label{Maxwell:eqn:v_times_operator}
\end{align}
to any vectorial quantity $v = (v_1,v_2,v_3)$.

\subsection{Acknowledgements} 
\label{intro:acknowledgements}
The authors thank L.~Esposito for sparking the interest in this topic. The foundation of this article was laid during the trimester program “Mathematical challenges of materials science and condensed matter physics”, and the authors thank the Hausdorff Research Institute for Mathematics for providing a stimulating research environment. Moreover, G.~D{.} gratefully acknowledges support by the Alexander von Humboldt Foundation and GNFM, “progetto giovani 2012”. M.~L{.} is supported by Deutscher Akademischer Austauschdienst. The authors also appreciate the useful comments and references provided by C.~Sparber and the two referees.

\section{The perturbed Maxwell operator} 
\label{Maxwell}
We will use this section to recall standard facts on the Maxwell operator \cite{Birman_Solomyak:L2_theory_Maxwell_operator:1987,Kuchment:math_photonic_crystals:2001} and introduce the main definitions and notions. This initial part is completed by a compendium of classical results in vector field analysis sketched in Appendix \ref{appendix:curl}.

\subsection{General properties of the Maxwell operator} 
\label{Maxwell:generic}
In order to identify the domain $\domain(\Maxwell_w)$ explicitly we start with the free case $\Mphys_{w = (1,1)} = \Rot$ which is reviewed in detail in Appendix \ref{appendix:ROT}. Assumption~\ref{intro:assumption:eps_mu_generic} on $w = (\eps,\mu)$ implies that $\Hil_w \simeq L^2(\R^3,\C^6)$ agree as Banach spaces and that $W$ defines a bounded operator with bounded inverse. Moreover, $\Rot \vert_{\Cont^{\infty}_{\mathrm{c}}}$ is a densely defined operators on $\Hil_w$ and $\Rot$ is its unique closed extension defined on the domain $\domain:= \domain(\Rot)$ (\cf eq. \eqref{Maxwell:eqn:domain_Maxwell}). Since, the graph norms $\norm{\cdot}_{\Mphys_w}$ and $\norm{\cdot}_{\Rot}$ are equivalent, this immediately implies
\begin{align}\label{Maxwell:eqn:domain_Maxwell2}
	\domain(\Mphys_w) = \domain = \bigl ( \ker \Div \cap H^1(\R^3,\C^6) \bigr ) \oplus \ran \Grad
	, 
\end{align}
because $\Mphys_w \vert_{\Cont^{\infty}_{\mathrm{c}}} = W \, \Rot \vert_{\Cont^{\infty}_{\mathrm{c}}}$ is closable and its \emph{unique} closure is the product of the bounded operator $W$ and $(\Rot,\domain)$. 

The weighted  scalar products \eqref{intro:eqn:relation_weighted_unweighted_scalar_product} also implies $\Mphys_w$ is not only closed but also symmetric, and thus, selfadjoint: for all $\Psi , \Phi \in \domain$, we have 
\begin{align*}
	\bscpro{\Psi}{\Mphys_w \Phi}_w &= \bscpro{\Psi}{W^{-1} \, W \, \Rot \, \Phi}_{L^2(\R^3,\C^6)} 
	= \bscpro{\Rot \, \Psi}{\Phi}_{L^2(\R^3,\C^6)} 
	\\
	&= \bscpro{W^{-1} \, W \, \Rot \, \Psi}{\Phi}_{L^2(\R^3,\C^6)} 
	= \bscpro{\Mphys_w \Psi}{\Phi}_w 
	. 
\end{align*}
The weights in the scalar products imply that the Helmholtz-Hodge-Weyl-Leray decomposition of the domain \eqref{Maxwell:eqn:domain_Maxwell2} is no longer orthogonal with respect to $\scpro{\cdot \,}{\cdot}_w$. However, Theorem~\ref{appendix:curl:thm:Leray_decomposition} readily generalizes to the case with weights and yields an orthogonal splitting 
\begin{align}
	\Hil_w = \Jphys_w \orthsum \Gphys
	\label{Maxwell:eqn:H_w_orth_splitting}
\end{align}
where we identify the \emph{physical} (or \emph{transversal}) subspace
\begin{align}
	\Jphys_w = \ker \bigl ( \Div \, W^{-1} \bigr ) 
	= \left \{ \Psi \in \Hil_w \; \; \vert \; \; \Div \bigl ( W^{-1} \Psi \bigr ) = 0 \right \} 
	= W \, \Jphys 
\end{align}
and the \emph{unphysical} (or \emph{longitudinal}) subspace 
\begin{align}
	\Gphys = \ran \Grad 
	= \left \{ \Psi = \Grad \, \varphi \in \Hil_w \; \; \vert \; \; \varphi \in L^2_{\mathrm{loc}}(\R^3,\C^2) \right \} 
	= \ker \Rot 
	. 
\end{align}
We also call $\Gphys$ the space of \emph{zero modes}, because $\Gphys = \ker \Rot$ coincides with $\ker \Mphys_w$ as $W$ has a bounded inverse. From the first equation of \eqref{intro:eqn:relation_weighted_unweighted_scalar_product} we conclude that $\Jphys_w = \Gphys^{\perp_w}$ is the $\scpro{\cdot \,}{\cdot}_w$-orthogonal complement to $\Gphys$. We will denote the orthogonal projections onto $\Jphys_w$ and $\Gphys$ with $\Pphys_w$ and $\Qphys_w$. For later reference, we summarize these facts into a 
\begin{thm}[\cite{Birman_Solomyak:L2_theory_Maxwell_operator:1987}]\label{Maxwell:thm:selfadjointness}
	Suppose Assumption~\ref{intro:assumption:eps_mu_generic} on $\eps$ and $\mu$ is satisfied. 
	\begin{enumerate}[(i)]
		\item The Maxwell operator $\Mphys_w$ equipped with the \emph{$(\eps,\mu)$-independent} domain 
		\begin{align*}
			\domain &= \bigl ( \domain \cap H^1(\R^3,\C^6) \bigr ) \directsum \ran \Grad 
			= \bigl ( \ker \Div \cap H^1(\R^3,\C^6) \bigr ) \directsum \Gphys 
		\end{align*}
		defines a selfadjoint operator on $\Hil_w$, and $H^1(\R^3,\C^6)$ and $\Cont^{\infty}_{\mathrm{c}}(\R^3,\C^6)$ are cores. 
		\item The Maxwell operator $\Mphys_w = \Mphys_w \vert_{\Jphys_w} \orthsum 0 \vert_{\Gphys}$ is block diagonal with respect to the \emph{$(\eps,\mu)$-dependent} orthogonal decomposition of $\Hil_w = \Jphys_w \orthsum \Gphys$. In this decomposition, the domain splits into 
		\begin{align*}
			\domain = \bigl ( \domain \cap \Jphys_w \bigr ) \orthsum \Gphys 
			\, . 
		\end{align*}
		\item The restrictions of $\Mphys_w$ to $\Jphys_w$ or $\Gphys$ again define selfadjoint operators, and thus, the dynamics $\e^{- \ii t \Mphys_w}$ leave $\Jphys_w$ and $\Gphys$ invariant. 
	\end{enumerate}
\end{thm}
With the exception of the explicit computation of the domain, all of this is contained in \cite[Lemma~2.2]{Birman_Solomyak:L2_theory_Maxwell_operator:1987}. 

We have mentioned the significance of admitting \emph{complex} vector fields in the introduction (\cf Footnote~\ref{intro:footnote:complex_fields}), and the question arises whether we can construct solutions by evolving $\Psi \in \Hil_w$ in time and then taking real and imaginary part of $\Psi(t) = \e^{- \ii t \Mphys_w} \Psi$. This question will be crucial as to why usually one needs to consider “counter-propagating waves” whose frequencies $\pm \omega$ differ by a sign. So let $(C \Psi)(x) := \overline{\Psi(x)}$, $\Psi \in L^2(\R^3,\C^N)$, be component-wise complex conjugation; for simplicity, we shall always use the same symbol independently of $N \in \N$. \emph{If $\eps$ and $\mu$ are real,} then the weights commute with $C$, and 
%
\begin{align*}
	\bigl ( C \Mphys_w C \Psi \bigr )^E &= 
	C \bigl ( + \ii \, \eps^{-1}(\hat{x}) \, \nabla_x^{\times}  \bigr ) C \psi^H 
	= - \ii \, \eps^{-1}(\hat{x}) \, \nabla_x \times \psi^H
\end{align*}
as well as an analogous computation for the other component of $\Mphys_w \Psi$ imply 
\begin{align}
	C \, \Mphys_w \, C = - \Mphys_w 
	\, . 
	\label{Maxwell:eqn:particle-hole_symmetry}
\end{align}
Consequently, the spectra for Maxwell operators with real weights are symmetric with respect to reflections at $0$; the same holds for all spectral components. 
\begin{thm}\label{Maxwell:thm:particle_hole_symmetry}
	Suppose Assumption~\ref{intro:assumption:eps_mu_generic} on the weights $\eps$ and $\mu$ is satisfied, and assume in addition that they are real. Then equation~\eqref{Maxwell:eqn:particle-hole_symmetry} holds and thus the spectra $\sigma (\Mphys_w) = - \sigma (\Mphys_w)$ and $\sigma_{\sharp} (\Mphys_w) = - \sigma_{\sharp}(\Mphys_w)$, $\sharp = \mathrm{pp} , \; \mathrm{ac} , \; \mathrm{sc}$, are symmetric with respect to reflections about the origin $0 \in \R$. 
\end{thm}
In case $\eps$ and $\mu$ have non-trivial complex offdiagonal entries, the weights no longer commute with complex conjugation, and \eqref{Maxwell:eqn:particle-hole_symmetry} as well as the above theorem do not hold. 
\begin{remark}
	Symmetries of type~\eqref{Maxwell:eqn:particle-hole_symmetry}, \ie \emph{anti-}unitary operators which map $\Mphys_w$ onto $- \Mphys_w$, are known in the physics literature as \emph{particle-hole symmetries} or \emph{PH symmetries} for short \cite{Altland_Zirnbauer:superconductors_symmetries:1997,Schnyder_Ryu_Furusaki_Ludwig:classification_topological_insulators:2008}. However, as many physicists and mathematicians consider the second-order equation $\partial_t^2 \Psi = - \Mphys_w^2 \Psi$ because it is block-diagonal, the PH symmetry for $\Mphys_w$ is replaced by a \emph{time-reversal symmetry} for the second-order equation. Ordinary Schrödinger operators $H = - \Delta_x + V$ on the other hand possess time-reversal symmetry, $C \, H \, C = H$. Discrete symmetries which square to $\pm \id$ have been classified systematically for topological insulators (\cf Table~II in \cite{Schnyder_Ryu_Furusaki_Ludwig:classification_topological_insulators:2008}); the presence of the PH symmetry means that $\Mphys_w$ is in \emph{symmetry class D} (provided there are no other symmetries). According to general results on the topological classification of band insulators (aka periodic operators), one expects that D-type operators in dimension $d = 2$ admit protected states parametrized by $\Z$-valued topological invariants (\cf Table~I in \cite{Schnyder_Ryu_Furusaki_Ludwig:classification_topological_insulators:2008}). This suggests there is an analog of the quantum Hall effect in $2$-dimensional photonic crystals \cite{Raghu_Haldane:quantum_Hall_effect_photonic_crystals:2008}. In contrast, for topological invariants to exist in $d = 3$, additional symmetries appear to be necessary (\eg $\eps = \mu$ or $\eps$ and $\mu$ have a common center of inversion); the presence of PH symmetry alone seems to prevent the formation of topologically protected states. Certainly, a direct proof for the Maxwell operator establishing the existence ($d = 2$) or absence ($d = 3$) of topological invariants would be an interesting avenue to explore. 
\end{remark}
%

\subsection{Slow modulation of the Maxwell operator} 
\label{Maxwell:perturbed}
One of the key differences between Maxwell and Schrödinger operators is that perturbations are \emph{multiplicative} rather than \emph{additive.} Given material weights $\eps$ and $\mu$ (which verify Assumption \ref{intro:assumption:eps_mu_generic}), we define their slow modulations $(\eps_{\lambda} , \mu_{\lambda})$ to be of the form \eqref{intro:eqn:slow_modulation_material_constants}. Assumption~\ref{Maxwell:assumption:modulation_functions} for the modulation functions $(\tau_\varepsilon , \tau_\mu )$ ensures that also $(\eps_{\lambda} , \mu_{\lambda})$ satisfy Assumption~\ref{intro:assumption:eps_mu_generic} because they are again bounded away from $0$ and $+ \infty$. 

We denote the $\lambda$-dependence of the weights with $w(\lambda) = (\eps_{\lambda},\mu_{\lambda})$ and define shorthand notation for the $\lambda$-dependent family of Hilbert spaces, projections and Maxwell operators by setting 
\begin{align*}
	&\Hil_{\lambda} := \Hil_{w(\lambda)} 
	\, , 
	&&\Jphys_{\lambda} := \Jphys_{w(\lambda)}
	&&&&\text{(spaces)}& \\
	&\Mphys_{\lambda} := \Mphys_{w(\lambda)} 
	\, , 
	&&\Pphys_{\lambda} := \Pphys_{w(\lambda)}
	\, , 
	&&\Qphys_{\lambda} := \Qphys_{w(\lambda)}
	&&\text{(operators)}
	\, . & 
\end{align*}
Similarly, we will denote the scalar product and norm of $\Hil_{\lambda}$ by $\scpro{\cdot \,}{\cdot}_{\lambda}$ and $\norm{\cdot}_{\lambda}$. 

To compare these operators for different values of $\lambda$, we will represent them on a \emph{common, $\lambda$-independent} Hilbert space: the scaling operator 
\begin{align}
	S(\lambda \hat{x}) : \Hil_{\lambda} \longrightarrow \Hil_0 
	, 
	\quad 
	S(\lambda \hat{x}) = \left (
	\begin{matrix}
		\tau_{\eps}^{-1}(\lambda \hat{x}) & 0 \\
		0 & \tau_{\mu}^{-1}(\lambda \hat{x}) \\
	\end{matrix}
	\right )
	\, , 
\end{align}
is a unitary since it is surjective and preserves scalar products. The Maxwell operator in this new representation can be calculated explicitly: for instance, the upper-right matrix element of $\Maxwell_{\lambda}$ transforms to 
\begin{align*}
	\tau_{\eps}^{-1}(\lambda \hat{x}) \, &\bigl ( - \tau_{\eps}^2(\lambda \hat{x}) \, \eps^{-1}(\hat{x}) \, (- \ii \nabla_x )^{\times} \bigr ) \, \tau_{\mu}(\lambda \hat{x}) = 
	\\
	&= - \tau_{\eps}(\lambda \hat{x}) \, \tau_{\mu}(\lambda \hat{x}) \, 
	\Bigl ( \eps^{-1}(\hat{x}) \, (- \ii \nabla_x)^{\times} + \lambda \, \eps^{-1}(\hat{x}) \, \bigl (- \ii \nabla_x \ln \tau_{\mu} \bigr )^{\times}(\lambda \hat{x}) \Bigr ) 
	\, , 
\end{align*}
and if we introduce the functions $\tau(\lambda x) := \tau_{\eps}(\lambda x) \, \tau_{\mu}(\lambda x)$ and 
\begin{align*}
	\derln(\lambda x) &:= \left (
	\begin{matrix}
		0 & + \ii \bigl ( \nabla_x \ln \tau_{\mu} \bigr )^{\times}(\lambda x) \\
		- \ii  \bigl (\nabla_x \ln \tau_{\eps} \bigr )^{\times}(\lambda x) & 0 \\
	\end{matrix}
	\right )
	\, , 
\end{align*}
we can write the Maxwell operator as 
\begin{align}
	M_{\lambda} :\negmedspace &= S(\lambda \hat{x}) \, \Mphys_{\lambda} \, S(\lambda \hat{x})^{-1}
	= M_0 + \lambda \, M_1 
	\notag \\
	&
	= \tau(\lambda \hat{x}) \, \Mper + \lambda \, \tau(\lambda \hat{x}) \, W \, \derln(\lambda \hat{x}) 
	\, . 
	\label{Maxwell:eqn:M_lambda_independent_rep}
\end{align}
As a product of bounded multiplication operators, $M_1$ is an element of $\mathcal{B}(\Hper)$. 

The regularity of $\tau_{\eps}$ and $\tau_{\mu}$ also ensures the domain is preserved. 
\begin{lem}
	$S(\lambda \hat{x})$ maps $\domain$ bijectively onto itself. 
\end{lem}
This means all of the operators, $\Mper$, $\Maxwell_{\lambda}$ and $M_{\lambda}$, have the same $\lambda$-independent domain $\domain$ and cores (\eg $H^1(\R^3,\C^6)$) -- even though the splitting of the domain into physical and unphysical subspaces depends on $\lambda$. We denote the invariant subspaces 
\begin{align*}
	J_{\lambda} := S(\lambda \hat{x}) \, \Jphys_{\lambda}
	\, , 
	\qquad\qquad
	G_{\lambda} := S(\lambda \hat{x}) \, \Gphys 
\end{align*}
of $M_{\lambda}$ with regular letters instead of bold letters, and in the same vein, the corresponding projections are 
\begin{align*}
	P_{\lambda} := S(\lambda \hat{x}) \, \Pphys_{\lambda} \, S(\lambda \hat{x})^{-1} 
	\, , 
	\qquad\qquad
	Q_{\lambda} := S(\lambda \hat{x}) \, \Qphys_{\lambda} \, S(\lambda \hat{x})^{-1} 
	\, . 
\end{align*}
For $\lambda = 0$, the $\lambda$-independent representation coincides with the physical representation since $S(\lambda \hat{x}) \vert_{\lambda = 0} = \id_{\Hil_0}$ reduces to the identity by Assumption~\ref{Maxwell:assumption:modulation_functions}, and we have $J_0 = \Jphys_0$ and $G_0 = \Gphys$ for the subspaces, as well as $P_0 = \Pphys_0$ and $Q_0 = \Qphys_0$ for the corresponding projections. 

The unitarity of $S(\lambda \hat{x})$ and Theorem~\ref{Maxwell:thm:selfadjointness} imply $\Hil_0 = J_{\lambda} \orthsum G_{\lambda}$ is a $\lambda$-dependent decomposition of $\Hil_0$ into $\scpro{\cdot \,}{\cdot}_0$-orthogonal subspaces which are invariant under the dynamics $\e^{- \ii t M_{\lambda}}$. 

\section{Properties of the periodic Maxwell operator} 
\label{periodic}
Photonic crystals are materials where the unperturbed material weights $(\eps,\mu)$ are periodic with respect to a lattice 
\begin{align*}
	\Gamma := \mathrm{span}_{\Z} \{ e_1 , e_2 , e_3 \} 
	\cong \Z^3 
	\, , 
\end{align*}
and henceforth, we shall always make the following 
\begin{assumption}[Photonic crystal]\label{periodic:assumption:periodic_eps_mu}
	Suppose that $\eps$ and $\mu$ are $\Gamma$-periodic and satisfy Assumption~\ref{intro:assumption:eps_mu_generic}. 
\end{assumption}
The lattice periodicity suggests we borrow the language of crystalline solids \cite{Grosso_Parravicini:solid_state_physics:2003}: we can decompose vectors $x = y + \gamma$ in real space $\R^3 \cong \WS \times \Gamma$ into a component $y$ which lies in the so-called Wigner--Seitz cell $\WS$ and a lattice vector $\gamma \in \Gamma$. Whenever convenient we will identify this fundamental cell $\WS$ with the $3$-dimensional torus $\T^3$. 

Given a lattice $\Gamma$, then there is a canonical way to decompose momentum space $\R^3 \cong \BZ \times \Gamma^*$: here, the \emph{dual lattice} $\Gamma^* = \mathrm{span}_{\Z} \{ e_1^* , e_2^* , e_3^* \}$ is generated by the family of vectors which are defined through the relations $e_j \cdot e_n^* = 2 \pi \, \delta_{jn}$, $j , n = 1 , 2 , 3$. The standard choice of fundamental cell 
\begin{align*}
	\BZ := \Bigl \{ \mbox{$\sum_{j = 1}^3$} \alpha_j \, e_j^* \in \R^3 \; \big \vert \; \alpha_1 , \alpha_2 , \alpha_3 \in [-\nicefrac{1}{2},+\nicefrac{1}{2}) \Bigr \} 
\end{align*}
is called (first) Brillouin zone, and elements $k \in \BZ$ are known as \emph{crystal momentum}.

\subsection{The Zak transform} 
\label{periodic:Zak}
The lattice-periodicity of $\eps$ and $\mu$ sugests to use a Fourier basis: for any $\C^N$-valued Schwartz function $\Psi \in \Schwartz(\R^3,\C^N)$ we define the \emph{Zak transform} \cite{Zak:dynamics_Bloch_electrons:1968} evaluated at $k \in \R^3$ and $y \in \R^3$ as 
\begin{align}
	( \Zak \Psi )(k,y) := \sum_{\gamma \in \Gamma} \e^{- \ii k \cdot (y + \gamma)} \, \Psi(y + \gamma) 
	\, . 
\end{align}
The Zak transform is a variant of the Bloch-Floquet transform with the following periodicity properties: 
\begin{align*}
	( \Zak \Psi )(k , y - \gamma) &= ( \Zak \Psi )(k , y) 
	&& 
	\gamma \in \Gamma 
	\\
	( \Zak \Psi )(k - \gamma^* , y) &= \e^{+ \ii \gamma^* \cdot y} ( \Zak \Psi )(k , y) 
	&& 
	\gamma^* \in \Gamma^* 
\end{align*}
In other words, $\Zak \Psi$ is a $\Gamma$-periodic function in $y$ and periodic up to a phase in $k$. The Schwartz functions are dense in $\Hper$, so 
\begin{align*}
	\Zak : \Hper \longrightarrow L^2_{\mathrm{eq}}(\R^3,\HperT) 
	\cong L^2(\BZ) \otimes \HperT 
\end{align*}
extends to a unitary map between $\Hper$ and the $L^2$-space of equivariant functions in $k$ with values in $\HperT := L^2_{\eps}(\T^3 , \C^3) \orthsum L^2_{\mu}(\T^3 , \C^3)$, 
\begin{align}
	L^2_{\mathrm{eq}} ( \R^3 , \HperT ) := \Bigl \{ 
	\Psi \in L^2_{\mathrm{loc}} ( \R^3 , \HperT ) \; \big \vert \; 
	\Psi(k - \gamma^*) = \e^{+ \ii \gamma^* \cdot \hat{y}} \Psi(k) \mbox{ a.~e. $\forall \gamma^* \in \Gamma^*$}
	\Bigr \} 
	\, , 
\end{align}
which is equipped with the scalar product 
\begin{align*}
	\scpro{\Psi}{\Phi}_{\mathrm{eq}} &:= \int_{\BZ} \dd k \, \bscpro{\Psi(k)}{\Phi(k)}_{\HperT}
\end{align*}
where 
\begin{align*}
	\bscpro{\Psi(k)}{\Phi(k)}_{\HperT} :& \negmedspace= \int_{\T^3} \dd y \; \psi^E(k,y)\cdot \eps(y) \phi^E(k,y) 
	\, + \\
	&\qquad \qquad 
	+ \int_{\T^3} \dd y \; \psi^H(k,y)\cdot  \mu(y) \phi^H(k,y) 
	\, . 
\end{align*}
Due to the (quasi-)periodicity of Zak transformed functions, they are uniquely determined by the values they take on $\BZ \times \T^3$. 

To see how the Maxwell operator transforms when conjugating it with $\Zak$, we compute the Zak representation of its building block operators positions $\hat{x}$ and momentum $- \ii \nabla_x$ (which are equipped with the obvious domains): 
\begin{align}
	\Zak \, \hat{x} \, \Zak^{-1} &= \ii \nabla_k 
	\label{periodic:eqn:Zak_position}
	\\
	\Zak \, (- \ii \nabla_x) \, \Zak^{-1} &= \id_{L^2(\BZ)} \otimes (- \ii \nabla_y) + \hat{k} \otimes \id_{\HperT} 
	\equiv - \ii \nabla_y + \hat{k} 
	\label{periodic:eqn:Zak_momentum}
\end{align}
The common domains of the components $\ii \partial_{k_j}$ and $- \ii \partial_{y_j} + \hat{k}_j$ Zak transform to $L^2_{\mathrm{eq}} ( \R^3 , \HperT ) \cap H^1_{\mathrm{loc}} \bigl ( \R^3 , \HperT  \bigr )$ and 
\begin{align}
	\Zak H^1(\R^3,\C^6) = L^2_{\mathrm{eq}} \bigl ( \R^3 , H^1(\T^3,\C^6) \bigr ) 
	\cong L^2(\BZ) \otimes H^1(\T^3,\C^6)
	\, . 
	\label{periodic:eqn:Zak_H1}
\end{align}
Note that the position operator in Zak representation does not factor, unless we consider $\Gamma$-periodic functions $\eps$, 
\begin{align}
	\Zak \, \eps(\hat{x}) \, \Zak^{-1} = \id_{L^2(\BZ)} \otimes \eps(\hat{y}) 
	\equiv \eps(\hat{y}) 
	\, . 
	\label{periodic:eqn:Zak_position_per}
\end{align}
Operators $\mathbf{A} : \domain(A) \longrightarrow \Hper$ which commute with lattice translations, \eg operators of the form \eqref{periodic:eqn:Zak_momentum}, \eqref{periodic:eqn:Zak_position_per} or the periodic Maxwell operator, fiber in $k$, 
\begin{align*}
	\mathbf{A}^{\Zak} = \Zak \, \mathbf{A} \, \Zak^{-1} = \int_{\BZ}^{\oplus} \dd k \, \mathbf{A}(k) 
	\, , 
\end{align*}
and the fiber operators at $k \in \R^3$ and $k - \gamma^*$, $\gamma^* \in \Gamma^*$, are unitarily equivalent, 
\begin{align}
	\mathbf{A}(k - \gamma^*) = \e^{+ \ii \gamma^* \cdot \hat{y}} \, \mathbf{A}(k) \, \e^{- \ii \gamma^* \cdot \hat{y}} 
	\, , 
	\label{periodic:eqn:equivariance_condition}
\end{align}
Operator-valued functions $k \mapsto \mathbf{A}(k)$ which satisfy \eqref{periodic:eqn:equivariance_condition} are called \emph{equivariant}. It is for this reason that it suffices to consider all objects only for $k \in \BZ$ and extend them by equivariance if necessary. 

\subsection{Analytic decomposition of the fiber Hilbert space} 
\label{periodic:fibration_spaces}
Clearly, $\Qper$ and $\Pper$ also commute with lattice translations, and thus, the Zak transform yields a fiber decomposition into 
\begin{align*}
	\Qper^{\Zak} := \Zak \, \Qper \, \Zak^{-1} 
	= \int^{\oplus}_{\BZ} \dd k \, \Qper(k)
	\, , \qquad 
	\Pper^{\Zak} := \Zak \, \Pper \, \Zak^{-1} 
	= \int^{\oplus}_{\BZ} \dd k \, \Pper(k) 
	\, . 
\end{align*}
These fibrations also identify physical and unphysical subspaces of the fiber Hilbert space 
\begin{align*}
	\HperT = \Jper(k) \orthsum \Gper(k) 
\end{align*}
for each $k \in \BZ$ where $\Gper(k) = \ran \Qper(k)$ and $\Jper(k) = \ran \Pper(k)$. A priori, all we know is that this fibration is \emph{measurable} in $k$. However, we are interested in the \emph{analyticity} properties of the fiber projections. Figotin and Kuchment have recognized that $k \mapsto \Qper(k)$ and thus also $k \mapsto \Pper(k)$ are not analytic at $k \in \Gamma^*$ \cite{Figotin_Kuchment:band_gap_structure_periodic_dielectric_acoustic_media:1996}. The purpose of this section is to define \emph{regularized} projections $k \mapsto \Qreg(k)$ and $k \mapsto \Preg(k)$ which are analytic on \emph{all} of $\R^3$. These regularized projections enter crucially in the proof on the existence of ground state bands (Theorem~\ref{intro:thm:band_picture}~(iii)). 
\begin{lem}\label{periodic:lem:regularized_projections}
	\begin{enumerate}[(i)]
		\item The orthogonal projections $k \mapsto \Qper(k)$ and $k \mapsto \Pper(k)$ onto unphysical and physical subspace are analytic on $\R^3 \setminus \Gamma^*$. 
		\item The regularized orthogonal projections $k \mapsto \Qreg(k)$ and $k \mapsto \Preg(k)$ are analytic on \emph{all} of $\R^3$. Moreover, $\Preg(\gamma^*) = \Pper(\gamma^*)$ and $\Qreg(\gamma^*) = \Qper(\gamma^*)$ holds for all $\gamma^* \in \Gamma^*$. 
		\item $\dim \bigl ( \Gper(k) \cap \Jreg(k) \bigr ) = 2$ for all $k \in \R^3 \setminus \Gamma^*$ 
	\end{enumerate}
\end{lem}
Essentially, the idea for the definition of $\Qreg(k)$ is already contained in the proofs of Lemma~51 and Corollary~52 of  \cite{Figotin_Kuchment:band_gap_structure_periodic_dielectric_acoustic_media:1996}, so we will briefly sketch the construction of $\Qper(k)$ and then proceed to define $\Qreg(k)$. 

Assume from now on that $k \in \BZ$. The idea is to use the fact that $\Gper(k) := \mathrm{ran}_0 \, \Grad(k)$ and define an auxiliary projection $\widetilde{Q}_0(k) = \Grad(k) \, T(k)$ with range $\Gper(k)$ as the product of the operator 
\begin{align*}
	\Grad(k) = (\nabla_y + \ii k , \nabla_y + \ii k) : H^1(\T^3,\C^2) \longrightarrow \HperT 
	.  
\end{align*}
which depends analytically on $k \in \R^3$ and its left-inverse $T(k)$. Such a left-inverse exists if and only if $\Grad(k)$ is injective, and \emph{if} it exists, it is also bounded \cite[p.~52]{Figotin_Kuchment:band_gap_structure_periodic_dielectric_acoustic_media:1996} and analytic in $k$ \cite[Theorem~4.4]{Zaidenberg_Krain_Kuchment_Pankov:Banach_bundles:1975}. Note that the closedness of $\mathrm{ran}_0 \, \Grad(k) = \Grad(k) \, H^1(\T^3,\C^2)$ for $k \neq 0$ follows from the boundedness of $T(k)$. 

One can check that for $k \neq 0$, the operator $\Grad(k)$ is injective while for $k = 0$ there are zero modes, 
\begin{align*}
	Z(\T^3,\C^2) := \biggl \{ y \mapsto \left (
	\begin{matrix}
		\beta^E \\
		\beta^H \\
	\end{matrix}
	\right ) \quad \Big \vert \quad \left (
	\begin{matrix}
		\beta^E \\
		\beta^H \\
	\end{matrix}
	\right ) \in \C^2 \biggr \} 
	= \ker \Grad(0) 
	\, . 
\end{align*}
Consequently, the projection $\widetilde{Q}_0(k) = \Grad(k) \, T(k)$ can only be defined in this fashion for $k \neq 0$, and there is a point of non-analyticity at $k = 0$, because $\ran \Grad(0)$ is “smaller” by two dimensions than $\Gper(k)$, $k \neq 0$. 

Even though $\widetilde{Q}_0(k)$ need not be an orthogonal projection (the proofs in \cite{Allan:holomorphic_vector-valued_functions:1967} and \cite{Zaidenberg_Krain_Kuchment_Pankov:Banach_bundles:1975} only make reference to the Banach algebra structure), these arguments show that $\Gper(k) = \ran \Qper(k) = \ran \widetilde{Q}_0(k)$ depends analytically on $k$ away from $\Gamma^*$. Thus, the unique \emph{orthogonal} projection $\Qper(k)$ onto $\Gper(k)$ necessarily also depends analytically on $k \in \R^3 \setminus \Gamma^*$. 

The behavior of $\Grad(k)$ at $k = 0$ suggests to define the \emph{regularized} unphysical space as 
\begin{align*}
	\Greg(k) := \mathrm{ran}_0 \, \Grad(k) \vert_{H^1_{\mathrm{reg}}} 
\end{align*}
where the closed subspace 
\begin{align*}
	H^1_{\mathrm{reg}}(\T^3,\C^2) :& \negmedspace = \Bigl \{ \varphi = (\varphi^E,\varphi^H) \in H^1(\T^3,\C^2) \; \, \big \vert \; \, \bscpro{1}{\varphi^{\sharp}}_{L^2(\T^3)} = 0 , \; \sharp = E , H \Bigr \} 
	\\
	&= Z(\T^3,\C^2)^{\perp} \cap H^1(\T^3,\C^2)
\end{align*}
consists of all $H^1$-functions orthogonal to the constant functions. Now $\Grad(k) \vert_{H^1_{\mathrm{reg}}}$ \emph{is} injective for all $k \in \BZ$, and by modifying the estimates on \cite[p.~52]{Figotin_Kuchment:band_gap_structure_periodic_dielectric_acoustic_media:1996} we deduce there exists an \emph{analytic} bounded left-inverse $T_{\mathrm{reg}}(k)$ for all $k \in \BZ$. Hence, the composition
\begin{align*}
	k \mapsto \widetilde{\Qreg}(k) := \Grad(k) \vert_{H^1_{\mathrm{reg}}} \, T_{\mathrm{reg}}(k) 
\end{align*}
defines a projection onto $\Greg(k)$ that depends analytically on $k$ for all of $\BZ$, including $k = 0$; again, the boundedness of $T_{\mathrm{reg}}(k)$ implies $\Greg(k)$ is a closed subset of $\HperT$. By the same arguments as above, the uniquely determined \emph{orthogonal} projection $\Qreg(k)$ onto $\Greg(k)$ inherits the analyticity of $\widetilde{\Qreg}(k)$ \cite[Theorem~6.35]{Kato:perturbation_theory:1995}. At $k = 0$, this regularized projection coincides with the usual one, $\Qreg(0) = \Qper(0)$, as their ranges

\begin{align}
	\Greg(0) = \ran \Grad(0) \vert_{H^1_{\mathrm{reg}}} 
	= \ran \Grad(0) 
	= \Gper(0)
	\label{periodic:eqn:Gper_0_equals_Greg_0}
\end{align}
are the same (this also proves that $\Gper(0)$ is closed). Moreover, $k \mapsto \Qreg(k)$ has a unique extension by equivariance (\cf \eqref{periodic:eqn:equivariance_condition}) to all of $\R^3$. 

Now the analyticity of the orthogonal projection 
\begin{align*}
	\Preg(k) := \id_{\HperT} - \Qreg(k) 
\end{align*}
onto the $\scpro{\cdot \,}{\cdot}_{\HperT}$-orthogonal complement 
\begin{align*}
	\Jreg(k) := \Greg(k)^{\perp}
\end{align*}
follows from the analyticity of $k \mapsto \Qreg(k)$. 
\medskip

\noindent
Before we prove (iii), it is instructive to juxtapose the decomposition $\HperT = \Jper(k) \orthsum \Gper(k)$ with the regularized decomposition 
\begin{align*}
	\HperT = \Jreg(k) \orthsum \Greg(k) 
\end{align*}
for the special case $\Mper = \Rot$, \ie $\eps = 1 = \mu$. The difference between the two is how the constant functions $y \mapsto (\alpha^E,\alpha^H) \in \C^6$, are distributed amongst them: for $k \neq 0$ only \emph{certain} constant functions belong to $\Jper(k)$, 
\begin{align*}
	y \mapsto \left (
	\begin{matrix}
		\alpha^E \\
		\alpha^H \\
	\end{matrix}
	\right ) \in \Jper(k) 
	\quad \Longleftrightarrow \quad 
	\Div(k) \left (
	\begin{matrix}
		\alpha^E \\
		\alpha^H \\
	\end{matrix}
	\right ) 
	= - \ii \, \left (
	\begin{matrix}
		k \cdot \alpha^E \\
		k \cdot \alpha^H \\
	\end{matrix}
	\right ) = \left (
	\begin{matrix}
		0 \\
		0 \\
	\end{matrix}
	\right )
	, 
\end{align*}
while for $k = 0$ \emph{all} constant functions are elements of $\Jper(0)$ and the physical subspace “grows” by $2$ dimensions at the expense of $\Gper(0)$. In contrast, the regularized physical subspace $\Jreg(k)$ contains all constant functions for \emph{all} values of $k$. We will now extend these arguments to the case of non-trivial weights $(\eps,\mu)$. 
\begin{proof}[Lemma~\ref{periodic:lem:regularized_projections}]
	We have already shown (i) and (ii) in the text preceding the lemma and it remains to prove (iii). Without loss of generality, we restrict ourselves to $k \in \BZ$. First of all, we note that the unphysical subspace 
	\begin{align*}
		\Gper(k) 
		&= \Biggl \{ 
		\sum_{\gamma^* \in \Gamma^*} \left (
		\begin{matrix}
			\beta^E(\gamma^*) \, (\gamma^* + k) \\
			\beta^H(\gamma^*) \, (\gamma^* + k) \\
		\end{matrix}
		\right ) \, \e^{+ \ii \gamma^* \cdot y} 
		\; \; \; \Big \vert \; \; \; 
		\Biggr . 
		\notag \\
		&\qquad \qquad \qquad 
		\Biggl . 
		\Bigl \{ \babs{\beta^{\sharp}(\gamma^*) \, \gamma^*} \Bigr \}_{\gamma^* \in \Gamma^*} \in \ell^2(\Gamma^*) 
		, \; 
		\sharp = E , H
		\Biggr \} 
	\end{align*}
	and the \emph{regularized} unphysical subspace 
	\begin{align}
		\Greg(k) 
		&= \Biggl \{ 
		\sum_{\gamma^* \in \Gamma^* \setminus \{ 0 \}} \left (
		\begin{matrix}
			\beta^E(\gamma^*) \, (\gamma^* + k) \\
			\beta^H(\gamma^*) \, (\gamma^* + k) \\
		\end{matrix}
		\right ) \, \e^{+ \ii \gamma^* \cdot y} 
		\; \; \; \Big \vert \; \; \; 
		\Biggr . 
		\notag \\
		&\qquad \qquad \qquad 
		\Biggl . 
		\Bigl \{ \babs{\beta^{\sharp}(\gamma^*) \, \gamma^*} \Bigr \}_{\gamma^* \in \Gamma^*} \in \ell^2(\Gamma^*) 
		, \; 
		\sharp = E , H
		\Biggr \} 
		. 
		\label{periodic:eqn:Greg_k}
	\end{align}
	coincide for $k = 0$, and we immediately deduce 
	\begin{align*}
		\dim \bigl ( \Gper(0) \cap \Jreg(0) \bigr ) &= \dim \bigl ( \Gper(0) \cap \Jper(0) \bigr ) 
		= 0
		\, . 
	\end{align*}
	Hence, we assume from now on $k \in \BZ \setminus \{ 0 \}$. That means, we can write the intersection as the regularized projection applied to a two-dimensional subspace, 
	\begin{align*}
		\Gper(k) \cap \Jreg(k) &= \Preg(k) 
		\left \{ 
		y \mapsto \left (
		\begin{matrix}
			\beta^E \, k \\
			\beta^H \, k \\
		\end{matrix}
		\right ) 
		\; \; \Big \vert \; \; \beta^E , \beta^H \in \C 
		\right \} 
		\, . 
	\end{align*}
	The image is again two-dimensional: if we write any $\Psi = \Psi_Q \orthsum \Psi_P \in \HperT$ as the sum of $\Psi_Q \in \Greg(k)$ and $\Psi_P \in \Jreg(k)$, then in view of equation~\eqref{periodic:eqn:Greg_k} the $\gamma^* = 0$ Fourier coefficient of $\psi_Q = \Qreg(k) \Psi$ necessarily has to vanish, $\widehat{\psi}_Q(0) = 0$. Thus, $\widehat{\psi}_P(0) = \widehat{\psi}(0)$ follows, and the map 
	\begin{align*}
		\C^2 \ni \left (
		\begin{matrix}
			\beta^E \\
			\beta^H \\
		\end{matrix}
		\right ) \mapsto \Preg(k) \left (
		\begin{matrix}
			\beta^E \, k \\
			\beta^H \, k \\
		\end{matrix}
		\right ) \in \Jreg(k) 
	\end{align*}
	is injective. That means $\dim \bigl ( \Gper(k) \cap \Jreg(k) \bigr ) = 2$ for $k \in \R^3 \setminus \Gamma^*$. 
\end{proof}
%

\subsection{Analyticity properties of the fiber Maxwell operator} 
\label{periodic:fiber_operators}
The Zak transform fibers the periodic Maxwell operator in crystal momentum, 
\begin{align}
	\Mper^{\Zak} :\negmedspace &= \Zak \, \Mper \, \Zak^{-1} 
	= \int_{\BZ}^{\oplus} \dd k \, \Mper(k) 
	\, . 
\end{align}
Each of the fiber operators 
\begin{align}
	\Mper(k) &= W \, \Rot(k) 
	= \left (
	\begin{matrix}
		0 & - \eps^{-1} \, (- \ii \nabla_y + k)^{\times} \\
		+ \mu^{-1} \, (- \ii \nabla_y + k)^{\times} & 0 \\
	\end{matrix}
	\right ) 
	, 
	\notag 
\end{align}
acts on a \emph{potentially $k$-depen\-dent} subspace $\domainT(k)$ of $\HperT$, and has a splitting into physical and unphysical part, $\Mper(k) = \Mper(k) \vert_{\Jper(k)} \directsum 0 \vert_{\Gper(k)}$. In any case, the selfadjointness of $\Mper$ on $\domain$ implies the selfadjointness of each fiber operator $\Mper(k)$ on $\domain(k)$. Since the domain of each fiber operator $\Mper(k)$ may depend on $k$, it is not obvious whether $k \mapsto \Mper(k)$ is analytic in $k$ even though the operator prescription is linear. 
\begin{prop}[Analyticity]\label{periodic:prop:analyticity_Mper}
	Suppose Assumption~\ref{periodic:assumption:periodic_eps_mu} on $\eps$ and $\mu$ holds. 
	\begin{enumerate}[(i)]
		\item The domain of selfadjointness 
		\begin{align}
			\domainT = \bigl ( \ker \Div(k) \cap H^1(\T^3,\C^6) \bigr ) \directsum \ran \Grad(k) 
			\label{periodic:eqn:domain_Mper_k}
		\end{align}
		of $\Mper(k)$ is independent of $k$. 
		\item The map $\R^3 \ni k \mapsto \Mper(k) \in \mathcal{B} ( \domainT , \HperT )$ is analytic. 
	\end{enumerate}
\end{prop}
\begin{proof}
	\begin{enumerate}[(i)]
		\item Since $H^1(\R^3,\C^6)$ is a core for $\Mper$ (Theorem~\ref{Maxwell:thm:selfadjointness}~(i)) and \eqref{periodic:eqn:Zak_H1}, we know that $H^1(\T^3,\C^6)$ is a common core of $\Mper(k)$ for all values of $k$. Moreover, combining equations \eqref{Maxwell:eqn:domain_Maxwell} and \eqref{periodic:eqn:Zak_H1} with the fact that $\Div$ and $\Grad$ also fiber in $k$ yields the decomposition of $\domainT$ as a $k$-dependent direct sum as given by \eqref{periodic:eqn:domain_Mper_k}. 
		
		The difference of the two fiber operators restricted to $H^1(\T^3,\C^6)$ extends to a bounded operator on all of $\HperT$, 
		\begin{align}
			\Mper(k) \vert_{H^1} - \Mper(k') \vert_{H^1} 
			&
			= W \, \left (
			\begin{matrix}
				0 & - (k - k')^{\times} \\
				+ (k - k')^{\times} & 0 \\
			\end{matrix}
			\right ) 
			\notag \\
			&
			=: \sum_{j = 1}^3 (k_j - k_j') \, \mathbf{A}_j 
			=: (k - k') \cdot \mathbf{A} 
			\, . 
			\label{periodic:eqn:definition_k_cdot_bfA}
		\end{align}
		Using $\bnorm{k \cdot \mathbf{A}}_{\mathcal{B}(\HperT)} = \sabs{k} \, \snorm{W}_{\mathcal{B}(\HperT)}$, it is straightforward to see that these graph norms of $\Mper(k)$ and $\Mper(0)$ are equivalent on $H^1(\T^3,\C^6)$, 
		\begin{align*}
			\bigl ( 1 + \abs{k} \, \snorm{W} \bigr )^{-1} \, \snorm{\Psi}_{\Mper(0)} \leq \snorm{\Psi}_{\Mper(k)} \leq \bigl ( 1 + \abs{k} \, \snorm{W} \bigr ) \, \snorm{\Psi}_{\Mper(0)} 
			\, . 
		\end{align*}
		The equivalence of the graph norms now implies that the domains, seen as completions of $H^1(\T^3,\C^6)$ with respect to these graph norms, are independent of $k$, 
		\begin{align*}
			\domainT(k) = \overline{H^1(\T^3,\C^6)}^{\norm{\cdot}_{\Mper(k)}} 
			= \overline{H^1(\T^3,\C^6)}^{\norm{\cdot}_{\Mper(0)}} 
			= \domainT(0)
			\, . 
		\end{align*}
		\item By (i) the domain $\domainT$ of each $\Mper(k)$ is independent of $k$, and thus the analyticity of the linear polynomial $k \mapsto \Mper(k)$ is trivial. 
	\end{enumerate}
\end{proof}
The fibration of $\Mper^{\Zak}$ can be used to extract a great deal of information on the spectra of $\Mper$ and $\Mper(k)$: 
\begin{thm}[Spectral properties]\label{periodic:thm:spectrum_Mper}
	Suppose Assumption~\ref{periodic:assumption:periodic_eps_mu} on $\eps$ and $\mu$ is satisfied. Then for any $k \in \R^3$ the following holds true: 
	\begin{enumerate}[(i)]
		\item $\displaystyle \sigma \bigl ( \Mper(k) \vert_{\Gper(k)} \bigr ) = \sigma_{\mathrm{ess}} \bigl ( \Mper(k) \vert_{\Gper(k)} \bigr ) 
		= \sigma_{\mathrm{pp}} \bigl ( \Mper(k) \vert_{\Gper(k)} \bigr ) 
		= \{ 0 \}$
		\item $\displaystyle \sigma \bigl ( \Mper(k) \vert_{\Jper(k)} \bigr ) 
		= \sigma_{\mathrm{disc}} \bigl ( \Mper(k) \vert_{\Jper(k)} \bigr )$ 
		\item $\displaystyle \sigma \bigl ( \Mper(k) \vert_{\Jreg(k)} \bigr ) 
		= \sigma_{\mathrm{disc}} \bigl ( \Mper(k) \vert_{\Jreg(k)} \bigr ) 
		= \sigma \bigl ( \Mper(k) \bigr )$ 
		\item $\displaystyle \sigma(\Mper) = \bigcup_{k \in \BZ} \sigma \bigl ( \Mper(k) \bigr ) 
		= \bigcup_{k \in \R^3} \sigma \bigl ( \Mper(k) \bigr )$
		\item $\displaystyle \sigma(\Mper) = \sigma_{\mathrm{ac}}(\Mper) \cup \sigma_{\mathrm{pp}}(\Mper)$
	\end{enumerate}
\end{thm}
\begin{proof}
	\begin{enumerate}[(i)]
		\item For any $\varphi \in \Cont^{\infty}_{\mathrm{c}}(\R^3,\C^2)$, the vector $\Grad(\varphi) \in \Gper$ is an element of the unphysical subspace, and thus we have found an eigenvector to $0$, 
		\begin{align*}
			\Mper(k) (\Zak \Psi)(k) = \bigl ( \Zak \Mper \Psi \bigr )(k) = 0 
			. 
		\end{align*}
		This means we have found a countably infinite family of eigenvectors, and we have shown (i). 
		\item According to Lemma~\ref{appendix:Maxwell:lem:properties_fibration_Rot}, $\bigl ( \Rot(k) \vert_{\Jrot(k)} - z \bigr )^{-1}$ is compact for all $k \in \R^3$ where $\Jrot(k) = \ker \Div(k)$ is the physical subspace for the free Maxwell operator. Because we can write $\bigl ( \Mper(k) \vert_{\Jper(k)} - z \bigr )^{-1}$ as a product of bounded operators and $\bigl ( \Rot(k) \vert_{\Jrot(k)} - z \bigr )^{-1}$ \cite[equation~(4.23)]{Sjoeberg_Engstroem_Kristensson_Wall_Wellander:Bloch-Floquet_Maxwell_homogenization:2005}, the resolvent of $\Mper(k) \vert_{\Jper(k)}$ is also compact. Thus, the spectrum of $\Mper(k) \vert_{\Jper(k)}$ is purely discrete. 
		\item This follows from (ii) and the observation that by Lemma~\ref{periodic:lem:regularized_projections}~(iii), $\Jper(k)$ and $\Jreg(k)$ differ by an at most $2$-dimensional subspace $\Jreg(k) \cap \Gper(k)$. 
		\item The proof is analogous to that of \cite[Corollary~57]{Figotin_Kuchment:band_gap_structure_periodic_dielectric_acoustic_media:1996}. 
		\item From (iv) we know that $\sigma(\Mper)$ can be written as the union of the spectra of the fiber operators $\Mper(k)$. Because these spectra $\sigma \bigl ( \Mper(k) \bigr ) = \bigl \{ \omega_n(k) \bigr \}_{n \in \Z}$ in turn can be expressed in terms of \emph{piecewise analytic} frequency band functions $k \mapsto \omega_n(k)$, $n \in \Z$ (\cf Theorem~\ref{intro:thm:band_picture}~(i)), $\sigma_{\mathrm{sc}}(\Mper)$ must be empty. 
	\end{enumerate}
\end{proof}
\begin{remark}[Absolute continuity of $\sigma \bigl ( \Mper \vert_{\Jper} \bigr )$]
	Unlike in the case of periodic Schrö\-ding\-er operators, it has not yet been proven that the spectrum of $\Mper \vert_{\Jper}$ is purely absolutely continuous. To show $\sigma \bigl ( \Mper \vert_{\Jper} \bigr ) = \sigma_{\mathrm{ac}} \bigl ( \Mper \vert_{\Jper} \bigr )$, all of the known proofs reduce the Maxwell operator to a possibly non-selfadjoint Schrödinger-type operator with magnetic field, and these transformations involve derivatives of $\eps$ and $\mu$ \cite{Morame:spectrum_purely_ac_periodic_Maxwell:2000,Suslina:ac_spectra_periodic_operators:2000,Kuchment_Levendorskii:spectrum_periodic_elliptic_operators:2001}. Hence, one needs additional regularity assumptions on $\eps$ and $\mu$; the best currently known are $\eps , \mu \in \Cont^1(\R^3)$ \cite[Section~7.4]{Kuchment_Levendorskii:spectrum_periodic_elliptic_operators:2001}. This means, even though it is widely expected that the spectrum is always purely absolutely continuous, flat bands (apart from $\omega \equiv 0$) currently cannot be excluded unless we make additional regularity assumptions on $\eps$ and $\mu$. 
\end{remark}
So far, most spectral and analytic properties mirror of $\Mper^{\Zak}$ those of periodic Schrödinger operators, but there are two important differences: (i) $\Mper$ is not bounded from below and (ii) in case of real weights the PH symmetry of the spectrum (\cf Theorem~\ref{Maxwell:thm:particle_hole_symmetry}) implies a symmetry for the frequency band spectrum (\cf Figure~\ref{periodic:fig:frequency_bands}). 

The first item in conjunction with the non-analyticity of $\Jper(k)$ at $k \in \Gamma^*$ potentially complicates the labeling of frequency bands. For simplicity, we solve this \emph{using} the band picture proven in Theorem~\ref{intro:thm:band_picture}: first of all, we know there exists an infinitely degenerate flat band $\omega_0(k) = 0$ associated to the unphysical states (\cf Theorem~\ref{periodic:thm:spectrum_Mper}~(i)). Moreover, it is easy to prove that $0$ is an eigenvalue of $\Mper(k) \vert_{\Jper(k)}$ if and only if $k \in \Gamma^*$. Away from $k \in \Gamma^*$, we repeat \emph{non-zero} eigenvalues $\omega_j(k)$ of $\Mper(k)$ according to their multiplicity, arrange them in non-increasing order and label positive (negative) eigenvalues with positive (negative) integers, \ie away from $k \in \Gamma^*$ we set 
\begin{align*}
	\ldots \leq \omega_{-2}(k) \leq \omega_{-1}(k) < \omega_0(k) = 0 < \omega_{1}(k) \leq \omega_{2}(k) \leq \ldots 
\end{align*}
Moreover, due to the analyticity of $k \mapsto \Mper(k)$, the eigenvalues depend on $k$ in a continuous fashion, and we extend this labeling by continuity to $k \in \Gamma^*$. This procedure yields a family $\bigl \{ k \mapsto \omega_n(k) \bigr \}_{n \in \Z}$ of $\Gamma^*$-periodic functions. 

Two types of bands are special: beside the zero mode band $\omega_0(k) = 0$ which is due to states in $\Gper(k)$, the \emph{ground state bands} are those of lowest frequency in absolute value: 
\begin{defn}[Ground state bands]\label{periodic:defn:ground_state}
	We call a frequency band $k \mapsto \omega_n(k)$ of $\Mper^{\Zak}$ a ground state band if and only if 
	\begin{enumerate}[(i)]
		\item $\displaystyle \lim_{k \rightarrow 0} \omega_n(k) = 0$ and 
		\item $\omega_n$ is not identically $0$ in a neighborhood of $k = 0$. 
	\end{enumerate}
	Moreover, we define $\Index_{\mathrm{gs}} \subset \Z$ to be the set of ground state band indices. 
\end{defn}
The ground state bands can be recovered from the space of zero modes 
\begin{align*}
	\mathrm{GS} := \ker \Mper(0) \cap \Jper(0) 
	\, . 
\end{align*}
using analytic continuation, and hence, also the use of $\Preg(k)$ instead of $\Pper(k)$ even though they coincide at $k = 0$. 
\begin{lem}[Ground state eigenfunctions at $k = 0$]\label{periodic:lem:ground_state_k_0}
	Suppose Assumption~\ref{periodic:assumption:periodic_eps_mu} holds true. 
	Then $\mathrm{GS} = \Preg(0) \bigl \{ y \mapsto a \; \vert \; a \in \C^6 \bigr \}$ is six-dimensional and any of its elements can be uniquely written as 
	\begin{align*}
		\Psi_a(y) := \bigl ( \Preg(0) a \bigr )(y) 
		&= \sum_{\gamma^* \in \Gamma} \widehat{\Psi}_a(\gamma^*) \, \e^{+ \ii \gamma^* \cdot y} 
	\end{align*}
	for some $a \in \C^6$. 
	The Fourier coefficients $\widehat{\Psi}_a(\gamma^*) = \bigl ( \widehat{\psi}_a^E(\gamma^*) , \widehat{\psi}_a^H(\gamma^*) \bigr )$ satisfy the following relations: 
	\begin{align}
		\widehat{\Psi}_a(0) &= a \in \C^6 
		\label{periodic:eqn:ground_state_eigenfunction}
		\\
		\widehat{\psi}_a^{\sharp}(\gamma^*) &\propto \gamma^* 
		&&
		\forall \gamma^* \in \Gamma^* \setminus \{ 0 \} 
		, \; 
		\sharp = E , H 
		\notag 
	\end{align}
\end{lem}
\begin{proof}
	First of all, seeing as $W$ is bounded with bounded inverse, $\ker \Mper(k) = \ker \Rot(k)$. A simple computation (\cf Lemma~\ref{appendix:Maxwell:lem:properties_fibration_Rot}) yields that any $\Psi \in \ker \Rot(0)$ is of the form 
	\begin{align*}
		\Psi = a + \Psi_G
	\end{align*}
	for some $a \in \C^6$ and $\Psi_G \in \Greg(0) = \Gper(0)$. Applying $\Preg(0)$ to both sides yields $\Preg(0) \Psi = \Psi_a$ and consequently, $\dim \mathrm{GS} \leq 6$. 
	
	From $\bigl [ \Mper(k) , \Preg(k) \bigr ] = 0$ we deduce $\Preg(0) \bigl \{ y \mapsto a \; \vert \; a \in \C^6 \bigr \} \subseteq \mathrm{GS}$. Moreover, in view of \eqref{periodic:eqn:Greg_k}, $y \mapsto a \in \C^6$ is an element of $\Greg(0) = \Gper(0)$ if and only if $a = 0$. Hence, $a \mapsto \Psi_a$ is injective and 
	\begin{align*}
		\dim \Preg(0) \bigl \{ y \mapsto a \; \vert \; a \in \C^6 \bigr \} 
		= \dim \mathrm{GS} 
		= 6 
		\, . 
	\end{align*}
	Finally, $\Preg(0) \bigl ( \Psi_a - a \bigr ) = 0$ means $\Psi_a - a \in \Greg(0)$, and thus using equation~\eqref{periodic:eqn:Greg_k} once more, we deduce $\widehat{\Psi}_a(\gamma^*) \propto \gamma^*$ and $\widehat{\Psi}_a(0) = a$. 
\end{proof}
We now proceed to the proof of Theorem~\ref{intro:thm:band_picture} which establishes the frequency band picture for periodic Maxwell operators (\cf Figure~\ref{periodic:fig:frequency_bands}). 
\begin{proof}[of Theorem \ref{intro:thm:band_picture}]
	\begin{enumerate}[(i)]
		\item Since $\Mper(k)$ is isospectral to its restriction $\Mper(k) \vert_{\Jreg(k)}$, let us consider the latter. First of all, $k \mapsto \omega_{0}(k) = 0$ is trivially analytic, we may assume $n \neq 0$. Thus, the analyticity away from band crossings follows from the purely discrete nature of the spectrum of $\Mper(k) \vert_{\Jreg(k)}$ (Theorem~\ref{periodic:thm:spectrum_Mper}~(iii)), the analyticity of $k \mapsto \Mper(k)$ (Proposition~\ref{periodic:prop:analyticity_Mper}~(ii)) and $k \mapsto \Preg(k)$ (Lemma~\ref{periodic:lem:regularized_projections}) combined with standard perturbation theory in the sense of Kato \cite{Kato:perturbation_theory:1995}. 
		
		The $\Gamma^*$-periodicity of $k \mapsto \omega_{n}(k)$ is deduced from the equivariance of $k \mapsto \Mper(k)$. 
		\item Now assume in addition that $\eps$ and $\mu$ are real. For $n = 0$, we trivially find $\omega_{0}(k) = 0 = - \omega_{0}(-k)$. So from now on, suppose $n \in \Z \setminus \{ 0 \}$. 
			
		One can check that upon Zak transform, the PH operator (complex conjugation) $C^{\Zak} := \Zak C \Zak^{-1}$ acts on elements of $\Psi \in L^2_{\mathrm{eq}}(\BZ,\HperT)$ as $\bigl ( C^{\Zak} \Psi \bigr )(k) = \overline{\Psi(-k)}$. Combined with $C^{\Zak} \, \Mper^{\Zak} = - \Mper^{\Zak} \, C^{\Zak}$ which follows from equation~\eqref{Maxwell:eqn:particle-hole_symmetry} since $\eps$ and $\mu$ are real, a straight-forward calculation shows that if $u_n(k)$ is an eigenfunction to $\omega_n(k)$, then $\bigl ( C^{\Zak} u_n \bigr )(k)$ is an eigenfunction to $- \omega_n(-k)$, and we have shown (ii). 
		\item To show (1), we will prove 
		\begin{align}
			0 \in \sigma \bigl ( \Mper(k) \vert_{\Jper(k)} \bigr ) 
			\; \Longleftrightarrow \; 
			0 \in \sigma \bigl ( \Rot(k) \vert_{\Jrot(k)} \bigr )
			\label{periodic:eqn:0_in_spec_Maxwell_equiv_0_in_spec_Rot}
		\end{align}
		first where $\Jrot(k) = \ker \Div(k)$ is the physical subspace of the free Maxwell operator, and since the spectrum of $\Rot$, 
		\begin{align*}
			\sigma \bigl ( \Rot(k) \vert_{\Jrot(k)} \bigr ) = \bigcup_{\gamma^* \in \Gamma^*} \bigl \{ \pm \sabs{k + \gamma^*} \bigr \} 
			, 
		\end{align*}
		is known explicitly (\cf Lemma~\ref{appendix:Maxwell:lem:properties_fibration_Rot}), this will prove $0 \in \sigma \bigl ( \Mper(k) \vert_{\Jper(k)} \bigr )$ if and only if $k \in \Gamma^*$. Hence, combined with Definition~\ref{periodic:defn:ground_state} this implies (1). 
		
		First of all, since the spectra $\sigma \bigl ( \Mper(k) \vert_{\Jper(k)} \bigr )$ are discrete for any $k \in \BZ$ (Theorem~\ref{periodic:thm:spectrum_Mper}~(ii)), we only need to consider the existence of eigenvectors. As the inverse of $W$ is bounded, the equations $\Mper(k) \Psi = 0$ and $\Rot(k) \Psi = 0$ are equivalent on the domain $\domainT$. We will now show that the existence of $\Psi_{\Mper} \in \Jper(k) \cap \domainT$ to $\Mper(k) \Psi_{\Mper} = 0$ is equivalent to the existence of a $\Psi_{\Rot} \in \ker \Div(k)$ which satisfies $\Rot(k) \Psi_{\Rot} = 0$. 
		
		Assume there exists an eigenvector $\Psi_{\Mper} \in \Jper(k) \cap \domainT$. Then by the direct decomposition of the domain $\domain = \ker \Div(k) \directsum \ran \Grad(k)$ implies we can uniquely write 
		\begin{align*}
			\Psi_{\Mper} = \Psi_{\Rot} + \Psi_{G}
		\end{align*}
		as the sum of $\Psi_{\Rot} \in \ker \Div(k)$ and $\Psi_G \in \Gper(k)$. Because the intersection $\Jper(k) \cap \Gper(k) = \{ 0 \}$ is trivial, we know $\Psi_{\Rot} \neq 0$. Hence, $\Psi_{\Rot}$ is an eigenvector of $\Rot(k)$, 
		\begin{align*}
			\Rot(k) \Psi_{\Rot}  = \Rot(k) \bigl ( \Psi_{\Mper} - \Psi_G \bigr ) 
			= 0 
			. 
		\end{align*}
		The converse statement is shown analogously and we have proven \eqref{periodic:eqn:0_in_spec_Maxwell_equiv_0_in_spec_Rot}. 
		
		Now we turn to (2): let us define $N := \sabs{\Index_{\mathrm{gs}}}$. By (ii), $N$ needs to be even. Due to \eqref{periodic:eqn:Gper_0_equals_Greg_0}, we may replace the physical subspace $\Jper(0)$ with its regularized version $\Jreg(0)$, and the six-dimensional space $\mathrm{GS}$ from Lemma~\ref{periodic:lem:ground_state_k_0} can also be defined in terms of $\Jreg(0)$. Thus, we already know $N \leq \dim \mathrm{GS} = 6$. Moreover, since $\dim \bigl ( \Gper(k) \cap \Jreg(k) \bigr ) = 2$ (Lemma~\ref{periodic:lem:regularized_projections}~(iii)) and $\Qper(k) \Jreg(k) \subset \Gper(k)$, the operator $\Mper(k) \vert_{\Jreg(k)}$ has a two-fold degenerate flat band $k \mapsto 0$ and we conclude that in fact, $N \leq 4$. 
		
		To show $N = 4$ and property~(2), we use {standard} analytic perturbation theory in the sense of Kato around the eigenvalue $0$: We have proven in (i) that all band functions are continuous, and thus if $\omega_n(0)=0$ there exists a neighborhood $V$ of $k = 0$ and a $\delta > 0$ such that $\sabs{\omega_n(k)} < \delta$ holds on $V$. Let us pick an orthonormal basis $\bigl \{ \Psi_1 , \ldots , \Psi_6 \bigr \}$ of $\mathrm{GS}$; according to Lemma~\ref{periodic:lem:ground_state_k_0}, each of these $\Psi_j$ is associated to a coefficient $a_{(j)} = \bigl ( a_{(j)}^E , a_{(j)}^H \bigr ) \in \C^6$, $j = 1 , \ldots , 6$ via \eqref{periodic:eqn:ground_state_eigenfunction}. Then $\Mper(0) \Psi_j = 0$ and \cite[equation~(2.40)]{Kato:perturbation_theory:1995} imply the ground state band functions $\{ \omega_n(k) \}_{n \in \Index_{\mathrm{gs}}}$ are approximately equal to the non-zero eigenvalues of the $k$-dependent matrix 
		\begin{align}
			k \cdot A := \Bigl ( \bscpro{\Psi_l}{k \cdot \mathbf{A} \Psi_j}_{\HperT} \Bigr )_{1 \leq l , j \leq 6} 
			\label{periodic:eqn:k_cdot_A}
		\end{align}
		where $k \cdot \mathbf{A} = \Mper(k) - \Mper(0)$ is explicitly given in equation~\eqref{periodic:eqn:definition_k_cdot_bfA} and $k \cdot A := \sum_{j=1}^3 k_j \, A_j$ involves the implicitly defined matrices $A_j$. For $a , b \in \C^6$, we can directly compute the scalar product: 
		\begin{align}
			\bscpro{\Psi_a}{k \cdot \mathbf{A} \Psi_b}_{\HperT} &= \scpro{\left (
			\begin{matrix}
				\psi_a^E \\
				\psi_a^H \\
			\end{matrix}
			\right )}{W \, \left (
			\begin{matrix}
				- k \times \psi_b^H \\
				+ k \times \psi_b^E \\
			\end{matrix}
			\right )}_{\HperT} 
			\notag \\
			&= k \cdot \int_{\T^3} \dd y \, \Bigl ( \overline{\psi_a^E(y)} \times \psi_b^H(y) - \overline{\psi_a^H(y)} \times \psi_b^E(y) \bigr ) \Bigr ) 
			\notag \\
			&
			= k \cdot \Bigl ( \overline{a^E} \times b^H - \overline{a^H} \times b^E \Bigr ) 
			\label{periodic:eqn:expval_matrix_k_cdot_bfA}
			\displaybreak[2]
			\\
			&= \scpro{\left (
			\begin{matrix}
				a^E \\
				a^H \\
			\end{matrix}
			\right )}{\left (
			\begin{matrix}
				0 & - k^{\times} \\
				+ k^{\times} & 0 \\
			\end{matrix}
			\right ) \left (
			\begin{matrix}
				b^E \\
				b^H \\
			\end{matrix}
			\right )}_{\C^6} 
			\label{periodic:eqn:expval_matrix_k_cdot_bfA_C6}
		\end{align}
		To arrive at the last line, we plug in the ansatz~\eqref{periodic:eqn:ground_state_eigenfunction} for the ground state function, use the orthogonality of the plane waves with respect to the standard scalar product on $L^2(\T^3)$ and exploit $\gamma^* \times \gamma^* = 0$. 
		
		Now let us define the invertible $6 \times 6$ matrix $\Lambda := \bigl ( a_{(1)} \; \vert \; \cdots \; \vert \; a_{(6)} \bigr )$ which maps the canonical basis $\bigl \{ v_{(1)} , \ldots , v_{(6)} \bigr \} \subset \C^6$ onto $\bigl \{ a_{(1)} , \ldots , a_{(6)} \bigr \}$. Then we can express the matrix elements of $k \cdot A$ in terms of $\Lambda$: 
		\begin{align}
			\bscpro{v_{(j)}}{k \cdot A \, v_{(n)}}_{\C^6} :& \negmedspace= \bigl ( k \cdot A \bigr )_{jn} 
			\notag \\
			&
			= \scpro{a_{(j)}}{\left (
			\begin{matrix}
				0 & - k^{\times} \\
				+ k^{\times} & 0 \\
			\end{matrix}
			\right ) a_{(n)}}_{\C^6} 
			\notag \\
			&= \scpro{v_{(j)}}{\Lambda^* \, \left (
			\begin{matrix}
				0 & - k^{\times} \\
				+ k^{\times} & 0 \\
			\end{matrix}
			\right ) \, \Lambda v_{(n)}}_{\C^6} 
			\label{periodic:eqn:matrix_elements_perturbation_expval_C6}
		\end{align}
		In view of equation~\eqref{periodic:eqn:expval_matrix_k_cdot_bfA}, the matrix elements possess an $SO(3)$ symmetry: if we define the action of $R \in SO(3)$ on $a \in \C^6$ by setting $R a := \bigl ( R a^E , R a^H \bigr )$, then equation~\eqref{periodic:eqn:expval_matrix_k_cdot_bfA} in conjunction with $R(v \times w) = Rv \times Rw$, $v , w \in \C^3$, yields 
		\begin{align}
			\bscpro{\Psi_a}{k \cdot \mathbf{A} \Psi_b}_{\HperT} 
			&= \bscpro{\Psi_{R a}}{(R k) \cdot \mathbf{A} \Psi_{R b}}_{\HperT} 
			. 
			\label{periodic:eqn:rotation_covariance_first-order_correction_ground_state}
		\end{align}
		Combining this symmetry with equation~\eqref{periodic:eqn:expval_matrix_k_cdot_bfA_C6}, we get 
		\begin{align*}
			\bigl ( &k \cdot A \bigr )_{jn} 
			= 
			\scpro{R \, \Lambda v_{(j)}}{\left (
			\begin{matrix}
				0 & + (R k)^{\times} \\
				- (R k)^{\times} & 0 \\
			\end{matrix}
			\right ) \, R \, \Lambda v_{(n)}}_{\C^6} 
			\\
			&= \scpro{v_{(j)}}{\bigl ( \Lambda^{-1} \, R \, \Lambda \bigr )^* \, \Lambda^* \, \left (
			\begin{matrix}
				0 & + (R k)^{\times} \\
				- (R k)^{\times} & 0 \\
			\end{matrix}
			\right ) \, \Lambda \, \bigl ( \Lambda^{-1} \, R \, \Lambda \bigr ) v_{(n)}}_{\C^6} 
		\end{align*}
		or, put more succinctly after replacing $R$ with $R^{-1}$ and $k$ with $Rk$, 
		\begin{align*}
			(Rk) \cdot A &= \bigl ( \Lambda^{-1} \, R^{-1} \, \Lambda \bigr )^* \, \bigl ( k \cdot A \bigr ) \, \bigl ( \Lambda^{-1} \, R^{-1} \, \Lambda \bigr ) 
			. 
		\end{align*}
		As the matrix $\Lambda^{-1} \, R^{-1} \, \Lambda$ is invertible, we deduce 
		\begin{align}
			\mathrm{rank} \bigl ( k \cdot A \bigr ) = \mathrm{rank} \bigl ( (R k) \cdot A \bigr ) 
			= \mathrm{rank} \bigl ( \lambda \, k \cdot A \bigr ) 
			\label{periodic:eqn:rotation_invariance_rank_perturbation_matrix}
		\end{align}
		holds for all $R \in SO(3)$ and $\lambda \in \C \setminus \{ 0 \}$, \ie the rank of the matrix $k \cdot A$ is independent of $k \neq 0$. In particular, it means that if $0 \in \sigma \bigl ( k_0 \cdot A \bigr )$ for some special $k_0 \neq 0$, then $0$ is an eigenvalue of \emph{all} matrices $k \cdot A$. 
		
		Now we will reduce this problem of $6 \times 6$ matrices to a problem of $3 \times 3$ matrices: first of all, \emph{any} basis $\bigl \{ v_{(j)} \bigr \}_{j = 1}^6$ of $\C^6$ gives rise to a basis $\bigl \{ \Psi_{v_{(j)}} \bigr \}_{j = 1}^6$ of $\mathrm{GS}$. In particular, if we take $\bigl \{ v_{(j)} \bigr \}_{j = 1}^6$ to be the canonical basis of $\C^6$, we can apply the Gram-Schmidt procedure to $\bigl \{ \Psi_{v_{(j)}} \bigr \}_{j = 1}^6$ and obtain a $\scpro{\cdot \,}{\cdot}_{\HperT}$-\emph{orthonormal} basis $\bigl \{ \Psi_{a_{(1)}} \bigr \}_{j = 1}^6$ of $\mathrm{GS}$ with coefficients $a_{(j)} = \bigl ( a_{(j)}^E , a_{(j)}^H \bigr ) \in \C^6$. 
		Due to the block structure of $W^{-1}$ that is also inherited by  $\bscpro{\Phi}{\Psi}_{\HperT} = \bscpro{\Phi}{W^{-1} \Psi}_{L^2(\T^3,\C^6)}$, the fact that $v_{(1)}^H = v_{(2)}^H = v_{(3)}^H = 0$ and $v_{(4)}^E = v_{(5)}^E = v_{(6)}^E = 0$ forces also the corresponding coefficients of the orthonormalized vectors to be $0$, 
		\begin{align*}
			a_{(1)}^H = a_{(2)}^H = a_{(3)}^H = 0
			\, , 
			\qquad 
			a_{(4)}^E = a_{(5)}^E = a_{(6)}^E = 0
			\, . 
		\end{align*}
		Moreover, $\bigl \{ a^E_{(1)} , a^E_{(2)}, a^E_{(3)} \bigr \}$ and $\bigl \{ a^H_{(4)} , a^H_{(5)}, a^H_{(6)} \bigr \}$ are two sets of linearly independent vectors in $\C^3$ with $a_{(1)}^E , a_{(4)}^H \propto (1,0,0)$.
		
		Thus, using  equation \eqref{periodic:eqn:expval_matrix_k_cdot_bfA}, one sees that the symmetric matrix $k \cdot A$ is purely block-offdiagonal and can be written in term of three $3 \times 3$ matrices $B = (B_1,B_2,B_3)$ as 
		\begin{align}
			k \cdot A =: \left (
			\begin{matrix}
				0 & k \cdot B \\
				\bigl ( k \cdot B \bigr )^* & 0 \\
			\end{matrix}
			\right )
			\, .
			\label{periodic:eqn:definition_k_cdot_B}
		\end{align}
		The block structure implies that 
		\begin{align}
			\mathrm{rank} \, \bigl ( k \cdot A \bigr ) = \mathrm{rank} \, \bigl ( k \cdot B \bigr ) + \mathrm{rank} \, \bigl ( k \cdot B \bigr )^{\ast}
			= 2 \, \mathrm{rank} \, \bigl ( k \cdot B \bigr ) 
			\, .
			\label{periodic:eqn:rank_kA_rank_kB}
		\end{align}
		Then in order to conclude that $\mathrm{rank} \, \bigl ( k \cdot A \bigr ) = 4$, we only need to show that $\mathrm{rank} \, \bigl ( k \cdot B \bigr ) = 2$. Since the result is independent of $k$, we pick $k_0 = (1,0,0)$ and use the basis obtained after Gram-Schmidt orthonormalizing $\bigl \{ \Psi_{v_{(1)}} , \ldots , \Psi_{v_{(6)}} \bigr \}$. Then $a_{(1)} \propto v_{(1)}$ and $a_{(4)} \propto v_{(4)}$ are non-trivial scalar multiples of $v_{(1)}$ and $v_{(4)}$, and consequently, one obtains again from~\eqref{periodic:eqn:expval_matrix_k_cdot_bfA}
		\begin{align*}
			k_0 \cdot B 
			&= \left (
			\begin{matrix}
				0 & 0 \\
				0 & k_0 \cdot \tilde{B}
			\end{matrix}
			\right )
		\end{align*}
		where the $2 \times 2$ matrix 
		\begin{align*}
			k_0 \cdot \tilde{B} = \left (
			\begin{matrix}
				k_0  \cdot \bigl ( \overline{a^E_{(2)}} \times a^H_{(5)} \bigr ) & k_0 \cdot \bigl ( \overline{a^E_{(2)}} \times a^H_{(6)} \bigr )  \\
				k_0  \cdot \bigl ( \overline{a^E_{(3)}} \times a^H_{(5)} \bigr ) & k_0 \cdot \bigl ( \overline{a^E_{(3)}} \times a^H_{(6)} \bigr )  \\
			\end{matrix}
			\right )
		\end{align*}
		has full rank, because $k_0 = v_{(1)}^E = v_{(4)}^H \propto a^E_{(1)} , a^H_{(4)}$ implies 
		\begin{align*}
			\mathrm{det} \bigl ( k_0 \cdot \tilde{B} \bigr ) 
			&= \overline{\Bigl ( k_0 \cdot \bigl ( a^E_{(2)} \times a^E_{(3)} \bigr ) \Bigr )} \; \Bigl ( k_0 \cdot \bigl ( a^H_{(5)} \times a^H_{(6)} \bigr ) \Bigr )
			\\
			&\propto \; \overline{\mathrm{det} \Bigl ( a^E_{(1)} \; \big \vert \; a^E_{(2)} \; \big \vert \; a^E_{(3)} \Bigr )} \; \mathrm{det} \Bigl ( a^H_{(4)} \; \big \vert \; a^H_{(5)} \; \big \vert \; a^H_{(6)} \Bigr ) \neq 0 
			\, .
		\end{align*}
		Hence, piecing together $\mathrm{rank} \, \bigl ( k_0 \cdot B \bigr ) = 2$ with equations~\eqref{periodic:eqn:rotation_invariance_rank_perturbation_matrix} and \eqref{periodic:eqn:rank_kA_rank_kB} yields that the degeneracy of the ground state bands is $4$. 
	\end{enumerate}
\end{proof}
%

\subsection{Comparison to existing literature} 
\label{periodic:comparison_literature}
Even though most of the results in this section are neither new nor surprising, we still feel they fill a void in the literature: To the best of our knowledge, it is the first time the most important fundamental properties of the fiber Maxwell operator $\Mper(k)$ are all proven rigorously in one place. Many of these are scattered throughout the literature, \eg various authors have proven the discrete nature of the spectrum of $\Mper(k)$ \cite{Figotin_Klein:localization_classical_waves_II:1997,Morame:spectrum_purely_ac_periodic_Maxwell:2000,Sjoeberg_Engstroem_Kristensson_Wall_Wellander:Bloch-Floquet_Maxwell_homogenization:2005} or have shown the non-analyticity of $k \mapsto \Pper(k)$ at $k = 0$ \cite{Figotin_Kuchment:band_gap_structure_periodic_dielectric_acoustic_media:1996}. Certainly there is no dearth of literature on the subject (see also \cite{Kuchment:math_photonic_crystals:2001,Joannopoulos_Johnson_Winn_Meade:photonic_crystals:2008} and references therein). However, most of these results are piecemeal: Some of them are contained in publications which do not really focus on the periodic Maxwell operator, but random Maxwell operators (\cite{Figotin_Klein:localization_classical_waves_I:1996,Figotin_Klein:localization_classical_waves_II:1997}, for instance). Other publications do not study $\Mper$ but rather operators associated to $\Mper^2$: since $\Mper^2$ is block-diagonal, it suffices to study a second-order equation for either $\mathbf{E}$ or $\mathbf{B}$, see \eg \cite{Figotin_Kuchment:band_gap_structure_periodic_dielectric_acoustic_media:1996,Figotin_Klein:localization_classical_waves_II:1997}. In the two-dimensional case, this leads to a \emph{scalar} equation where the right-hand side is a second-order operator \cite{Figotin_Kuchment:band_gap_structure_periodic_dielectric_acoustic_media:1996}. 

Nevertheless, one result is new, namely Theorem~\ref{intro:thm:band_picture}~(iii): even though the presence of ground state bands is heuristically well-understood, we provide rather simple and straight-forward proof. The $k \to 0$ limit is related in spirit to the \emph{homogenization limit} where the wavelength of the electromagnetic wave is large compared to the lattice spacing (see \eg \cite{Suslina:homogenization_Maxwell:2004,Suslina:homogenization_Maxwell:2005,Sjoeberg_Engstroem_Kristensson_Wall_Wellander:Bloch-Floquet_Maxwell_homogenization:2005,Birman_Suslina:homogenization_Maxwell:2007,Allaire_Palombaro_Rauch:diffractive_Bloch_wave_packets_Maxwell:2012} and references therein). On the one hand, many homogenization techniques yield much farther-reaching results, most notably effective equations for the dynamics (\eg \cite[Theorem~2.1]{Birman_Suslina:homogenization_Maxwell:2007}) while Theorem~\ref{intro:thm:band_picture}~(iii) only makes a statement about the behavior of the ground state frequency bands. On the other hand, compared to, say, \cite[Theorem~2.1]{Birman_Suslina:homogenization_Maxwell:2007} or \cite[Theorem~6.2]{Sjoeberg_Engstroem_Kristensson_Wall_Wellander:Bloch-Floquet_Maxwell_homogenization:2005}, computing the dispersion of the ground state bands for small $k$ seems much easier in our approach: given $\eps$ and $\mu$, the problem reduces to orthonormalizing $2 \times 3$ vectors numerically and solving an eigenvalue problem for an explicitly given $3 \times 3$ matrix $\sabs{k \cdot B}$ defined through \eqref{periodic:eqn:definition_k_cdot_B} with one known eigenvalue (namely $0$). Moreover, a proof of the fact that there are $4$ ground state bands also appears to be new, \eg in a recent publication this was stated as \cite[Conjecture~1]{Sjoeberg_Engstroem_Kristensson_Wall_Wellander:Bloch-Floquet_Maxwell_homogenization:2005}. Proving this fact, however, required a better insight into the nature of the singularity of $k \mapsto \Pper(k)$ at $k = 0$ and necessitated the introduction of a regularized projection $\Preg$. 

\section{$\Mphys_{\lambda}^{\Zak}$ and $M_{\lambda}^{\Zak}$ as $\Psi$DOs} 
\label{pseudo_Maxwell}
After expounding the properties of the periodic Maxwell operator, we proceed to the proof of Theorem~\ref{intro:thm:perturbed_Maxwell_psuedo}. The essential ingredient is a suitable interpretation of the usual Weyl quantization rule 
\begin{align}
	\Op_{\lambda}(f) := \frac{1}{(2\pi)^3} \int_{\R^3} \dd r' \int_{\R^3} \dd k' \; (\Fourier_{\sigma} f)(r',k') \; \e^{- \ii (k' \cdot (\ii \lambda \nabla_k) - r' \cdot \hat{k})} 
	\label{pseudo_Maxwell:eqn:Op}
\end{align}
where 
\begin{align*}
	(\Fourier_{\sigma} f)(r',k') := \frac{1}{(2\pi)^3} \int_{\R^3} \dd r' \int_{\R^3} \dd k' \; \e^{+ \ii (k' \cdot r - r' \cdot k)} \, f(r,k) 
\end{align*}
is the symplectic Fourier transform. The idea is to combine the point of view from \cite[Appendix~B]{Teufel:adiabatic_perturbation_theory:2003} and \cite[Section~2.2]{DeNittis_Lein:Bloch_electron:2009} with the fact that most results of standard pseudodifferential theory depend only on the \emph{Banach} structure of the spaces involved and not on the Hilbert structure. 

First of all, equation~\eqref{pseudo_Maxwell:eqn:Op} defines a $\Psi$DO for a large class of scalar \cite{Folland:harmonic_analysis_hase_space:1989,Hoermander:Weyl_calculus:1979,Kumanogo:pseudodiff:1981,Taylor:PsiDO:1981} and vector-valued functions \cite{Luke:operator-valued_PsiDOs:1972,Levendorskii:asymptotics_eigenvalue_distribution_DOs:1990}. For instance, if $f$ is a Hörmander symbol or order $m \in \R$ and type $\rho \in [0,1]$ taking values in the Banach space $(\mathcal{B},\norm{\cdot}_{\mathcal{B}})$, 
\begin{align}
	f \in \Hoer{m}(\mathcal{B}) :& \negmedspace= \left \{ 
	f \in \Cont^{\infty}(\R^6,\mathcal{B})
	\; \; \big \vert \; \; 
	\forall \alpha , \beta \in \N_0^3 : \snorm{f}_{m , \alpha \beta} < \infty
	\right \} 
	, 
\end{align}
where the seminorms $\bigl \{ \norm{\cdot}_{m , \alpha \beta} \bigr \}_{\alpha , \beta \in \N_0^3}$ are defined by 
\begin{align*}
	\snorm{f}_{m , \alpha \beta} := \sup_{(r,k) \in \R^6} \Bigl ( \sqrt{1 + k^2}^{\; -m + \sabs{\beta} \rho} \, \bnorm{\partial_r^{\alpha} \partial_k^{\beta} f(r,k)}_{\mathcal{B}} \Bigr )
	\, , 
\end{align*}
then \eqref{pseudo_Maxwell:eqn:Op} is defined as an oscillatory integral \cite{Hoermander:Fourier_integral_operators_1:1971}. The vector-valuedness of $f$ usually does not create any technical difficulties, most standard results readily extend to vector-valued symbols, \eg Caldéron-Vaillancourt-type theorems and the composition of Hörmander-type symbols (see \eg \cite{Luke:operator-valued_PsiDOs:1972,Gerard_Martinez_Sjoestrand:Bloch_electron:1991,Martinez_Sordoni:PsiDO_application_dynamics_molecules:2009} and \cite[Appendix~A]{Teufel:adiabatic_perturbation_theory:2003}). 

In our applications $\mathcal{B} = \mathcal{B}(\mathfrak{h}_1,\mathfrak{h}_2)$ will always be some Banach space of bounded operators between the Hilbert spaces $\mathfrak{h}_1$ and $\mathfrak{h}_2$ whose elements are $L^2$-functions on the torus, \eg $L^2(\T^3,\C^N)$, $\HperT$ or $\domainT$. As explained in \cite[Section~2.2.1]{DeNittis_Lein:Bloch_electron:2009}, when compared to the pseudodifferential calculus associated to $(- \ii \lambda \nabla_x , \hat{x})$, equation~\eqref{pseudo_Maxwell:eqn:Op} can be seen as an equivalent representation of the same underlying Moyal algebra \cite{Gracia_Bondia_Varilly:distributions_phasespace_1:1988,Gracia_Bondia_Varilly:distributions_phasespace_2:1988}. Hence, the usual formulas and results apply, and we may use standard Hörmander classes instead of the less common weighted Hörmander classes as in \cite{PST:effective_dynamics_Bloch:2003}.

\subsection{Equivariant $\Psi$DOs} 
\label{pseudo_Maxwell:equivariant}
The relevant Hilbert spaces, $\Zak \Hil_{\lambda}$ and $\Zak \Hper$, coincide with $L^2_{\mathrm{eq}} \bigl ( \R^3 , L^2(\T^3,\C^6) \bigr )$ as Banach spaces, and we are in the same framework as in \cite[Appendix~B]{Teufel:adiabatic_perturbation_theory:2003} and \cite[Section~2.2.2]{DeNittis_Lein:Bloch_electron:2009}. The building block operators are macroscopic position $\ii \lambda \nabla_k$ and crystal momentum $\hat{k}$ whose domains are dense in $L^2_{\mathrm{eq}}(\R^3,\HperT)$ (\cf Section~\ref{periodic:Zak}). 

Operators which fiber-decompose in Zak representation have the equivariance property~\eqref{periodic:eqn:equivariance_condition}, and thus $\Mper^{\Zak} : L^2_{\mathrm{eq}}(\R^3,\domainT) \longrightarrow L^2_{\mathrm{eq}}(\R^3,\HperT)$ defines a selfadjoint operator between Hilbert spaces of equivariant functions, for instance. This motivates the following 
\begin{defn}[Semiclassical symbols]\label{pseudo_Maxwell:defn:equivariant_symbols}
	Assume $\mathfrak{h}_j$, $j = 1 , 2$, are Hilbert spaces consisting of functions on $\T^3$. A map $f : [0,\lambda_0) \longrightarrow \Hoereq{m} \bigl ( \mathcal{B}(\mathfrak{h}_1,\mathfrak{h}_2) \bigr )$, $\lambda \mapsto f_{\lambda}$, is called a semiclassical equivariant symbol of order $m \in \R$ and weight $\rho \in [0,1]$, that is $f \in \SemiHoereq{m} \bigl ( \mathcal{B}(\mathfrak{h}_1,\mathfrak{h}_2) \bigr )$, if and only if 
	\begin{enumerate}[(i)]
		\item $f_{\lambda}(r,k - \gamma^*) = \e^{- \ii \gamma^* \cdot \hat{y}} \, f_{\lambda}(r,k) \, \e^{+ \ii \gamma^* \cdot \hat{y}}$ holds $\forall$ $(r,k) \in \R^6$, $\gamma^* \in \Gamma^*$ and 
		\item there exists a sequence $\{ f_n \}_{n \in \N_0}$, $f_n \in \Hoer{m - n \rho}$, such that for all $N \in \N_0$ 
		\begin{align*}
			\lambda^{-N} \left ( f_{\lambda} - \sum_{n = 0}^{N - 1} \lambda^n \, f_n \right ) \in \Hoer{m - N \rho} \bigl ( \mathcal{B}(\mathfrak{h}_1,\mathfrak{h}_2) \bigr )
		\end{align*}
		holds true uniformly in $\lambda$ in the sense that for any $N \in \N_0$ and $\alpha , \beta \in \N_0^3$, there exist constants $C_{N \alpha \beta} > 0$ so that the estimate
		\begin{align*}
			\norm{f_{\lambda} - \sum_{n = 0}^{N - 1} \lambda^n \, f_n}_{m , \alpha \beta} \leq C_{N \alpha \beta} \, \lambda^N 
		\end{align*}
		is satisfied for all $\lambda \in [0,\lambda_0)$. 
	\end{enumerate}
\end{defn}
Since $\Hoer{m} \bigl ( \mathcal{B}(\mathfrak{h}_1,\mathfrak{h}_2) \bigr )$ and $\Hoereq{m} \bigl ( \mathcal{B}(\mathfrak{h}_1,\mathfrak{h}_2) \bigr )$  are contained in the Moyal algebra \cite[Section~III]{Gracia_Bondia_Varilly:distributions_phasespace_1:1988}, the associated $\Psi$DOs extend from continuous maps between vector-valued Schwartz functions to continuous maps between vector-valued tempered distributions, 
\begin{align*}
	\Op_{\lambda} \left ( \Hoer{m} \bigl ( \mathcal{B}(\mathfrak{h}_1,\mathfrak{h}_2) \bigr ) \right ) 
	&\subset \Op_{\lambda} \left ( \Hoereq{m} \bigl ( \mathcal{B}(\mathfrak{h}_1,\mathfrak{h}_2) \bigr ) \right ) 
	\\
	&
	\subset \mathcal{L} \left ( \Schwartz(\R^3,\mathfrak{h}_1) , \Schwartz(\R^3,\mathfrak{h}_2) \right ) \cap \mathcal{L} \left ( \Schwartz'(\R^3,\mathfrak{h}_1) , \Schwartz'(\R^3,\mathfrak{h}_2) \right ) 
	. 
\end{align*}
Furthermore, one can easily check that equivariant $\Psi$DOs also preserve equivariance on the level of tempered distributions: let us define translations and multiplication with the phase $\e^{+ \ii \gamma^* \cdot \hat{y}}$ on  $\Schwartz'(\R^3,\mathfrak{h}_j)$, $j = 1 , 2$, by duality, \ie we set 
\begin{align*}
	\bigl ( L_{\gamma^*} F , \varphi \bigr )_{\Schwartz} := \bigl ( T , \varphi(\cdot + \gamma^*) \bigr )_{\Schwartz} 
	\, , 
	\qquad 
	\bigl ( \e^{- \ii \gamma^* \cdot \hat{y}} F , \varphi \bigr )_{\Schwartz} := \bigl ( T , \e^{+ \ii \gamma^* \cdot \hat{y}} \varphi \bigr )_{\Schwartz} 
	\, , 
\end{align*}
for all $\gamma^* \in \Gamma^* \subset \R^3$. The set of equivariant tempered distributions  $\Schwartz'_{\mathrm{eq}}(\R^3,\mathfrak{h}_j)$, $j = 1 , 2$, is comprised of those tempered distributions which satisfy 
\begin{align*}
	L_{\gamma^*} F = \e^{- \ii \gamma^* \cdot \hat{y}} F 
	\, . 
\end{align*}
Then \cite[Proposition~B.3]{Teufel:adiabatic_perturbation_theory:2003} states that 
\begin{align*}
	\Op_{\lambda}(f) : \Schwartz'_{\mathrm{eq}}(\R^3,\mathfrak{h}_1) \longrightarrow \Schwartz'_{\mathrm{eq}}(\R^3,\mathfrak{h}_2)
\end{align*}
holds for all $f \in \Hoereq{m} \bigl ( \mathcal{B}(\mathfrak{h}_1,\mathfrak{h}_2) \bigr )$. Consequently, the inclusion $L^2_{\mathrm{eq}}(\R^3,\mathfrak{h}_j) \subset \Schwartz'_{\mathrm{eq}}(\R^3,\mathfrak{h}_j)$ and the standard Caldéron-Vaillancourt theorem imply \cite[Proposition~B.5]{Teufel:adiabatic_perturbation_theory:2003}
\begin{align*}
	\Op_{\lambda} \left ( \Hoereq{0} \bigl ( \mathcal{B}(\mathfrak{h}_1,\mathfrak{h}_2) \bigr ) \right ) \subset \mathcal{B} \left ( L^2_{\mathrm{eq}}(\R^3 , \mathfrak{h}_1) \, , \, L^2_{\mathrm{eq}} (\R^3 , \mathfrak{h}_2) \right ) 
	. 
\end{align*}
Similarly, the Moyal product $\sharp$ which is implicitly defined through 
\begin{align*}
	\Op_{\lambda} (f \sharp g) := \Op_{\lambda}(f) \, \Op_{\lambda}(g) 
\end{align*}
extends as a bilinear, continuous map which respects equivariance, 
\begin{align}
	\sharp : \Hoereq{m_1} \bigl ( \mathcal{B}(\mathfrak{h}_1,\mathfrak{h}_2) \bigr ) \times \Hoereq{m_2} \bigl ( \mathcal{B}(\mathfrak{h}_2,\mathfrak{h}_3) \bigr ) \longrightarrow \Hoereq{m_1 + m_2} \bigl ( \mathcal{B}(\mathfrak{h}_1,\mathfrak{h}_3) \bigr ) 
	. 
	\label{pseudo_Maxwell:eqn:Weyl_composition_symbols}
\end{align}
%

\subsection{Extension to weighted $L^2$-spaces} 
\label{psuedo_Maxwell:weighted_L2}
We have seen that certain equivariant operator-valued functions define bounded $\Psi$DOs mapping between Hilbert spaces of equivariant $L^2$-functions. The fact $\mathcal{B}(\mathfrak{h}_1,\mathfrak{h}_2)$ only depends on the Banach space structure of $\mathfrak{h}_1$ and $\mathfrak{h}_2$ immediately implies 
\begin{align*}
	\mathcal{B} \bigl ( \Zak \domain \, , \, \Zak \Hil_{\lambda} \bigr ) = \mathcal{B} \Bigl ( L^2_{\mathrm{eq}}(\R^3,\domainT) \, , \, L^2_{\mathrm{eq}} \bigl ( \R^3 , L^2(\T^3,\C^6) \bigr ) \Bigr )
	\, , 
\end{align*}
for instance, and hence any $f \in \Hoereq{0} \bigl ( \mathcal{B} \bigl ( L^2(\T^3,\C^6) \bigr ) \bigr )$ uniquely defines a $\Psi$DO 
\begin{align}
	\Op_{\lambda}(f) : \Zak \Hil_{\lambda} \longrightarrow \Zak \Hil_{\lambda} 
	. 
	\label{pseudo_Maxwell:eqn:Op_lambda-dependent_spaces}
\end{align}
One only needs to be careful about taking adjoints: the adjoint operator crucially depends on the scalar product (see \eg the discussion of selfadjointness of $\Mphys_w$ in Section~\ref{Maxwell:generic}), but in our applications, properties such as selfadjointness are checked “by hand”. 

\subsection{Proof of Theorem~\ref{intro:thm:perturbed_Maxwell_psuedo}} 
\label{pseudo_Maxwell:proof}
Assumption~\ref{periodic:assumption:periodic_eps_mu} on the material weights $\eps$ and $\mu$ as well as Assumption~\ref{Maxwell:assumption:modulation_functions} placed on the modulation functions imply $\Hil_{\lambda}$ and $\Hper$ coincide with $L^2(\R^3,\C^6)$ as Banach spaces. Similarly, we have $\HperT = L^2(\T^3,\C^6)$ on the level of Banach spaces. This means, $\Zak \Hil_{\lambda}$ and $\Zak \Hper$ agree with $L^2_{\mathrm{eq}} \bigl ( \R^3 , L^2(\T^3,\C^6) \bigr )$ as normed vector spaces. 

Seeing as we can write $\Mphys_{\lambda}^{\Zak} = S(\ii \lambda \nabla_k)^{-2} \, \Mper^{\Zak}$, Theorem~\ref{intro:thm:perturbed_Maxwell_psuedo} follows from the following 
\begin{lem}
	Under the assumptions of Theorem~\ref{intro:thm:perturbed_Maxwell_psuedo}, the following two operators are semiclassical pseudodifferential operators: 
	\begin{enumerate}[(i)]
		\item $S(\ii \lambda \nabla_k)^{\pm 1} = \Op_{\lambda} \bigl ( S^{\pm 1} \bigr )$ where $S , S^{-1} \in \Hoermeq{0}{1} \left ( \mathcal{B} \bigl ( L^2(\T^3,\C^6) \bigr ) \right )$ 
		\item $\Mper^{\Zak} = \Op_{\lambda} \bigl ( \Mper(\, \cdot \,) \bigr )$ where $\Mper(\, \cdot \,) \in \Hoermeq{1}{1} \left ( \mathcal{B} \bigl ( \domainT , L^2(\T^3,\C^6) \bigr ) \right )$ 
	\end{enumerate}
\end{lem}
\begin{proof}
	\begin{enumerate}[(i)]
		\item The matrix $S(r)$ is block-diagonal with respect to $L^2(\T^3,\C^6) \cong L^2(\T^3,\C^3) \directsum L^2(\T^3,\C^3)$ and each block is proportional to the identity in $L^2(\T^3,\C^3)$. Due to the assumption on the modulation functions, we conclude 
		\begin{align*}
			S \in \Cont^{\infty}_{\mathrm{b}} \left ( \R^3 , \mathcal{B} \bigl ( L^2(\T^3,\C^6) \bigr ) \right ) \subset \Hoerm{0}{1} \left ( \mathcal{B} \bigl ( L^2(\T^3,\C^6) \bigr ) \right ) 
			. 
		\end{align*}
		Equivariance is trivial, because $S(\ii \lambda \nabla_k)$ commutes with $\e^{- \ii \gamma^* \cdot \hat{y}}$ and hence 
		\begin{align*}
			S(r) = \e^{+ \ii \gamma^* \cdot \hat{y}} \, S(r) \, \e^{- \ii \gamma^* \cdot \hat{y}}
		\end{align*}
		holds. Lastly, $S^{-1}$ has the same properties as $S$ since $\tau_{\eps}^{-1}$ and $\tau_{\mu}^{-1}$ also satisfy Assumption~\ref{Maxwell:assumption:modulation_functions}. This concludes the proof of (i). 
		\item By Proposition~\ref{periodic:prop:analyticity_Mper}, the map $k \mapsto \Mper(k)$ is linear (the domain is independent of $k$), and thus $\Hoerm{1}{1} \left ( \mathcal{B} \bigl ( \domainT , L^2(\T^3,\C^6) \bigr ) \right )$. Equivariance follows from equation~\eqref{periodic:eqn:equivariance_condition}, and thus we have shown (ii). 
	\end{enumerate}
\end{proof}
Seeing as $\Mper(\, \cdot \,)$ is linear, the asymptotic expansion of $\sharp$ terminates after two terms and the symbols of the Maxwell operators in the physical representation can be computed from 
\begin{align*}
	\Mphys^{\Zak} &= \Op_{\lambda} \bigl ( S^{-2} \sharp \Mper(\, \cdot \,) \bigr ) 
	=: \Op_{\lambda}(\pmb{\mathcal{M}}_{\lambda}) 
	. 
\end{align*}
That $\pmb{\mathcal{M}}_{\lambda}$ is an element of $\SemiHoermeq{1}{1} \left ( \mathcal{B} \bigl ( \domainT , L^2(\T^3,\C^6) \bigr ) \right )$ is implied by the composition properties of equivariant symbols \eqref{pseudo_Maxwell:eqn:Weyl_composition_symbols} and the preceding Lemma. This concludes the proof of Theorem~\ref{intro:thm:perturbed_Maxwell_psuedo}. 

Consequently, also the Maxwell operator in the auxiliary representation is a semiclassical $\Psi$DO, 
\begin{align*}
	M_{\lambda}^{\Zak} &= \Op_{\lambda} \bigl ( S \sharp \pmb{\mathcal{M}}_{\lambda} \sharp S^{-1} \bigr ) 
	= \Op_{\lambda} \bigl ( S^{-1} \sharp \Mper(\, \cdot \,) \sharp S^{-1} \bigr ) 
	=: \Op_{\lambda}(\Msymb_{\lambda}) 
	\, , 
\end{align*}
whose semiclassical symbol $\Msymb_{\lambda}$ is in the same symbol class. 
\begin{cor}\label{intro:cor:perturbed_Maxwell_psuedo}
	Under the assumptions of Theorem~\ref{intro:thm:perturbed_Maxwell_psuedo}, the Maxwell operator $M_{\lambda}^{\Zak} = \Op_{\lambda}(\Msymb_{\lambda})$ in the rescaled representation is the semiclassical pseudodifferential operator associated to 
	\begin{align*}
		\Msymb_{\lambda}(r,k) &= \tau(r) \, \Mper(k)
		- \lambda \, \tau(r) \, W \, 
		\left (
		\begin{matrix}
			0 & \tfrac{\ii}{2} \, \bigl ( \nabla_r \ln \nicefrac{\tau_{\eps}}{\tau_{\mu}} \bigr )^{\times}(r) \\
			\tfrac{\ii}{2} \, \bigl ( \nabla_r \ln \nicefrac{\tau_{\eps}}{\tau_{\mu}} \bigr )^{\times}(r) & 0 \\
		\end{matrix}
		\right )
	\end{align*}
	where $\tau(r) := \tau_{\eps}(r) \, \tau_{\mu}(r)$. 
	The function $\mathcal{M}_{\lambda} \in \SemiHoermeq{1}{1} \left ( \mathcal{B} \bigl ( \domainT , L^2(\T^3,\C^6) \bigr ) \right )$ is an equivariant semiclassical symbol in the sense of Definition~\ref{pseudo_Maxwell:defn:equivariant_symbols}.
\end{cor}

\renewcommand{\curl}{\mathbf{curl}}
\begin{appendix}
	\section{The $\curl$ operator and the $\Rot$ operator} 
	\label{appendix:curl}
	The aim of this Appendix is to clarify the meaning of the relation
	$\domain(\Rot) = \domain(\curl) \directsum \domain(\curl)$ used in Section~\ref{Maxwell:generic} in order to define the domain of the Maxwell operator. So to conclude our arguments from Section~\ref{Maxwell:generic}, we give a brief overview on the theory of the operators $\curl := \nabla_x^{\times}$ and $\Rot$. Many works have been devoted to the rigorous study of $\curl$ on $L^2(\Omega,\C^3)$ where $\Omega \subseteq \R^3$ can be a bounded \cite{Yoshida_Giga:spectra_rot:1990,Amrouche_Bernardi_Dauge_Girault:vector_potentials_3d_domains:1998,Hiptmair_Kotiuga_Tordeux:selfadjointness_curl:2012} or unbounded domain \cite{Picard:selfadjointness_curl:1998} whose boundary satisfies various regularity properties. A lot of related results are contained in standard texts on the Navier-Stokes equation \cite{Duvaut_Lions:inequalities_mechanics:1972,Foias_Temam:remarks_Navier_Stokes:1978,Girault_Raviart:finite_elements_Navier_Stokes:1986,Galdo:steady_state_Navier_Stokes:2011}. In this Appendix, we enumerate some elementary results for the special case $\Omega = \R^3$. The crucial result is the so-called \emph{Helm\-holtz-Hodge-Weyl-Leray decomposition} which leads to a decomposition of any $\psi \in L^2(\R^3,\C^3)$ into divergence and rotation-free component.

	\subsection{The gradient operator} 
	\label{appendix:curl:grad}
	The gradient operator is initially defined on the smooth functions with compact support by
	\begin{align}
		\grad : \Cont^{\infty}_{\mathrm{c}}(\R^3) \longrightarrow \Cont^{\infty}_{\mathrm{c}}(\R^3,\C^3) 
		, 
		\qquad 
		\grad \varphi := \left (
		\begin{matrix}
			\partial_{x_1} \varphi \\
			\partial_{x_2} \varphi \\
			\partial_{x_3} \varphi \\
		\end{matrix}
		\right ) 
		. 
		\label{eq:grad}
	\end{align}
	The operator $\grad$ is closable (any component $\partial_{x_j}$ is anti-symmetric) and its closure, still denoted with $\grad$, has domain $\domain(\grad) = H^1(\R^3)$ and trivial null space, $\ker \grad = \{ 0 \}$.

	\subsection{The divergence operator} 
	\label{appendix:curl:div}
	The second operator of relevance, the divergence 
	\begin{align}
		\div : \Cont^{\infty}_{\mathrm{c}}(\R^3,\C^3) \longrightarrow \Cont^{\infty}_{\mathrm{c}}(\R^3)
		, \qquad
		\div \, \psi := \sum_{j=1}^3\partial_{x_j} \psi_j 
		,
		\label{appendix:curl:eqn:div}
	\end{align}
	is also closable and its closure, still denoted with $\div$, has domain \cite[Section~1.2 and Theorem~1.1]{Temam:theory_Navier_Stokes:2001}
	\begin{align*}
		\domain(\div) := \overline{\Cont^{\infty}_{\mathrm{c}}(\R^3,\C^3)}^{\norm{\cdot}_{\div}} 
		= \bigl \{ \psi \in L^2(\R^3,\C^3) \; \; \vert \; \; \div \, \psi \in L^2(\R^3) \bigr \} 
		. 
	\end{align*}
	A relevant result is the \emph{Stokes formula} \cite[Theorem~1.2]{Temam:theory_Navier_Stokes:2001}, \ie we have 
	\begin{align*}
		X_{\psi}(\varphi) := \bscpro{\psi}{\grad \varphi}_{L^2(\R^3,\C^3)} + \bscpro{\div \, \psi}{ \varphi}_{L^2(\R^3)} = 0
	\end{align*}
	for all $\psi \in \domain(\div)$ and $\varphi \in H^1(\R^3)$. This follows mainly from the Cauchy-Schwarz inequality $ \babs{X_\psi(\phi)} \leqslant 2 \, \snorm{\psi}_{\div} \, \snorm{\phi}_{\grad}$. The above relation shows that $\div$ is the \emph{adjoint} of $-\grad$ and vice versa (\cf \cite{Picard:selfadjointness_curl:1998}). In this sense $\domain(\div)$ can be seen as the space of vector fields with weak divergence.

	\subsection{The rotor operator} 
	\label{appendix:curl:curl}
	Lastly, the 
	\begin{align}
		\curl : \Cont^{\infty}_{\mathrm{c}}(\R^3,\C^3) \longrightarrow \Cont^{\infty}_{\mathrm{c}}(\R^3,\C^3) 
		, \quad \quad 
		\curl \, \psi := \left (
		\begin{matrix}
			\partial_{x_2} \psi_3 - \partial_{x_3} \psi_2 \\
			\partial_{x_3} \psi_1 - \partial_{x_1} \psi_3 \\
			\partial_{x_1} \psi_2 - \partial_{x_2} \psi_1 \\
		\end{matrix}
		\right )
	\end{align}
	is essentially selfadjoint, and thus, uniquely extends to a selfadjoint operator whose domain 
	\begin{align}
		\domain(\curl) := \overline{\Cont^{\infty}_{\mathrm{c}}(\R^3,\C^3)}^{\norm{\cdot}_{\curl}} 
		= \left \{ \psi \in L^2(\R^3,\C^3) \; \; \vert \; \; \curl \, \psi \in L^2(\R^3,\C^3) \right \} 
		\label{appendix:curl:eqn:domain_curl} 
	\end{align}
	is the closure of the core with respect to the graph norm. The characterization of $\domain(\curl)$ by the second equality in \eqref{appendix:curl:eqn:domain_curl} is proven in a slightly more general context in \cite[Chapter~7, Lemma~4.1]{Duvaut_Lions:inequalities_mechanics:1972} (\cf also \cite[Definition~2.2]{Amrouche_Bernardi_Dauge_Girault:vector_potentials_3d_domains:1998} and \cite{Urban:wavelet_bases:2001}). By showing that the deficiency indices of $\curl$ are both $0$, \ie $\curl \, \psi = \pm \ii \, \psi$ has no non-trivial solutions, one deduces $\curl$ is indeed selfadjoint (\cf \cite{Chandrasekhar_Kendall:force_free_magnetic_fields:1957,Picard:selfadjointness_curl:1998}). A very interesting fact relates the domains of $\curl$ and $\div$, and the space $H^1(\R^3,\C^3)$: Theorem~2.5 of \cite{Amrouche_Bernardi_Dauge_Girault:vector_potentials_3d_domains:1998} states 
	\begin{align}
		\domain(\curl) \cap \domain(\div) = H^1(\R^3,\C^3) 
		\label{appendix:curl:eqn:intersection_domain_curl_div_equal_H1}
	\end{align}
	which follows from the identity 
	\begin{align}
		\norm{\psi}_{H^1(\R^3,\C^3)}^2 = \norm{\psi}_{L^2(\R^3,\C^3)}^2 + \norm{\curl \, \psi}_{L^2(\R^3,\C^3)}^2 + \norm{\div \, \psi}_{L^2(\R^3)}^2 
		. 
		\label{appendix:curl:eqn:H1_norm_in_terms_of_curl_div}
	\end{align}
	This decomposition of the $H^1(\R^3,\C^3)$-norm follows from integration by parts and the identity 
	\begin{align*}
		\bigl ( \curl \bigr )^2 = \grad \; \div - \Delta_x 
	\end{align*}
	on $\Cont^{\infty}_{\mathrm{c}}(\R^3,\C^3)$, and a simple density argument. Note that \eqref{appendix:curl:eqn:intersection_domain_curl_div_equal_H1} implies $\Cont^{\infty}_{\mathrm{c}}(\R^3,\C^3)$ and $H^1(\R^3,\C^3)$ are cores for both, $\div$ and $\curl$. 

	\subsection{The Helmholtz-Hodge-Weyl-Leray decomposition} 
	\label{appendix:curl:Leray_decomposition}
	For a more precise characterization of the domain $\domain(\curl)$ we need the \emph{Helmholtz-Hodge-Weyl-Leray decomposition} (see \cite[Chapter~I, Section~1.4]{Temam:theory_Navier_Stokes:2001},  \cite[Section~1.1]{Foias_Temam:remarks_Navier_Stokes:1978} and \cite[Section~III.1]{Galdo:steady_state_Navier_Stokes:2011}). Let us introduce the subspaces 
	\begin{align*}
		\mathbf{C}_{\sigma} :& \negmedspace= \left \{ \psi \in \Cont^{\infty}_{\mathrm{c}}(\R^3,\C^3) \; \; \vert \; \; \div \, \psi = 0 \right \}
		, 
		\qquad \quad
		\Jphys := \overline{\mathbf{C}_\sigma}^{\norm{\cdot}_{L^2(\R^3,\C^3)}} 
		. 
	\end{align*}
	\begin{thm}[Helmholtz-Hodge-Weyl-Leray decomposition]
	\label{appendix:curl:thm:Leray_decomposition}
		The space $ L^2(\R^3,\C^3)$ admits the following orthogonal decomposition
		\begin{align}
			L^2(\R^3,\C^3) = \Jphys \orthsum \Gphys 
			\label{appendix:curl:eqn:Leray_decomposition}
		\end{align}
		where $\Jphys \subset \domain(\div)$ is defined by
		\begin{align}
			\Jphys = \left \{ \psi \in L^2(\R^3,\C^3) \; \; \vert \; \; \div \, \psi = 0 \right \} 
			= \ker \div 
			\label{appendix:curl:eqn:definition_J}
		\end{align}
		and 
		\begin{align}
			\Gphys := \bigl \{ \psi \in L^2(\R^3,\C^3) \; \; \vert \; \; \psi = \grad \varphi , \; \varphi \in L^2_{\mathrm{loc}}(\R^3) \bigr \} 
			= \ran \grad 
			\, .
			\label{appendix:curl:eqn:definition_G}
		\end{align}
		Moreover, one has also the following characterization:
		\begin{align}
			\Jphys = \ker \div = \ran \curl 
			, \qquad \quad
			\Gphys = \ker \curl = \ran \grad
			\, .
			\label{appendix:curl:eqn:characterization_J_G_ker_ran}
		\end{align}
	\end{thm}
	\begin{proof}[Sketch]
		Equation \eqref{appendix:curl:eqn:definition_J} is proven in \cite[Chapter~I, Theorem~1.4, eq. (1.34)]{Temam:theory_Navier_Stokes:2001}. The inclusion $\Jphys \subset \domain(\div)$ follows from the observation that the norms $\norm{\cdot}_{L^2(\R^3,\C^3)}$ and $\norm{\cdot}_{\div}$ coincide on $\mathbf{C}_{\sigma}$. 
		
		The definition of $\Gphys$ as gradient fields (first equality) has been shown in \cite[Chapter I, Theorem 1.4, eq. (1.33) and Remark 1.5]{Temam:theory_Navier_Stokes:2001}. The closedness of $\Gphys$, and thus, the second equality is discussed in the proof of \cite[Lemma~2.5]{Picard:selfadjointness_curl:1998}. (According to our choice of convention in Section~\ref{intro:notation}, $\ran \grad$ is the closure of $\mathrm{ran}_0 \, \grad = \grad \, H^1(\R^3)$, and for an example of $\varphi \in L^2_{\mathrm{loc}}(\R^3)\setminus H^1(\R^3)$ such that $\grad \varphi \in L^2(\R^3,\C^3)$ we refer to \cite[Note~2, pg.~156]{Galdo:steady_state_Navier_Stokes:2011}.) 
		
		The proofs of the two remaining equalities in \eqref{appendix:curl:eqn:characterization_J_G_ker_ran} can be found in \cite[Theorem~1.1]{Picard:selfadjointness_curl:1998}. 
		
		We remark that in case of the vector fields on all of $\R^3$, the space of harmonic vector fields $H_N := \ker \div \cap \ker \curl = \{ 0 \}$ is the trivial vector space, because $\Delta \psi = 0$ has no non-trivial solutions on $L^2(\R^3,\C^3)$. This concludes the proof of \eqref{appendix:curl:eqn:Leray_decomposition}. 
	\end{proof}
	\begin{remark}\label{appendix:curl:remark:nomenclature}
		According to the standard nomenclature $\Jphys$ is known as the space of the \emph{solenoidal} or \emph{transversal} vector fields while $\Gphys$ is the space of the \emph{irrotational} or \emph{longitudinal} vector fields. The orthogonal projection $\Pphys : L^2(\R^3,\C^3) \longrightarrow \Jphys $ is called \emph{Leray projection}. The identification $\Jphys = \ran \curl$ implies that $\curl : \Jphys \longrightarrow \Jphys$ and this is enough for $[\Pphys , \curl] = 0$.
	\end{remark}
	Theorem \ref{appendix:curl:thm:Leray_decomposition} has two immediate consequences: The first is the \emph{Helmholtz splitting}, meaning each $\psi \in L^2(\R^3,\C^3)$ can be uniquely decomposed into a \emph{stream field} $\phi \in \domain(\curl)$ and the gradient of a \emph{potential function} $\varphi \in L^2_{\mathrm{loc}}(\R^3)$, 
	\begin{align*}
		\psi = \curl \, \phi + \grad \varphi 
		, 
	\end{align*}
	where $\curl \, \phi$ and $\grad \varphi$ are mutually orthogonal. The second is the content of the following 
	\begin{cor}[Domain of $\curl$]\label{appendix:curl:cor:domain_curl}
		The domain $\domain(\curl)$ of the operator $\curl$ admits the following splitting
		\begin{align}
			\domain(\curl) &= \bigl ( \Jphys \cap \domain(\curl) \bigr ) \orthsum \Gphys 
			\notag \\
			&
			= \bigl ( \Jphys \cap H^1(\R^3,\C^3) \bigr ) \orthsum \Gphys 
			\notag \\
			&
			= \bigl ( \ker \div \cap H^1(\R^3,\C^3) \bigr ) \orthsum \ker \curl 
			\notag \\
			&
			= \bigl ( \ker \div \cap H^1(\R^3,\C^3) \bigr ) \orthsum \ran \grad 
			.
			\label{appendix:curl:eqn:splitting_domain_curl}
		\end{align}
	\end{cor}
	\begin{proof}
		Theorem \ref{appendix:curl:thm:Leray_decomposition} implies $\domain(\curl) = \bigl ( \Jphys \cap\domain(\curl) \bigr ) \orthsum \Gphys$ since $\Gphys \subset \domain(\curl)$. Moreover, relation \eqref{appendix:curl:eqn:intersection_domain_curl_div_equal_H1} and $\Jphys = \ker \div$ lead to $\Jphys \cap \domain(\curl) = \bigl ( \Jphys \cap \domain(\div) \bigr ) \cap \domain(\curl) = \Jphys \cap H^1(\R^3,\C^3)$.
	\end{proof}
	%

	\subsection{The $\Rot$ operator} 
	\label{appendix:ROT}
	The block structure displayed in equation \eqref{intro:eqn:Rot} implies $\Rot$ defines a selfadjoint operator on $\domain(\Rot) = \domain(\curl) \orthsum \domain(\curl)$ where $\domain(\curl)$ is the domain of the rotation operator $\curl$ as given in Corollary~\ref{appendix:curl:cor:domain_curl}. The splitting \eqref{appendix:curl:eqn:splitting_domain_curl} of $\domain(\curl)$ carries over to $\Rot$, namely 
	\begin{align}
		\domain := \domain(\Rot) = \bigl ( \ker \Div \cap H^1(\R^3,\C^6) \bigr ) \orthsum \ran \Grad
		, 
		\label{Maxwell:eqn:domain_Maxwell}
	\end{align}
	where $\Div := \div \directsum \div$ and $\Grad := \grad \oplus \grad$ consist of two copies of $\div$ and $\grad$ which are defined as in Appendix~\ref{appendix:curl}, and $\ran \Grad$ is the closure of $\mathrm{ran}_0 \, \Grad$. 

	The splitting of the domain \eqref{Maxwell:eqn:domain_Maxwell} is motivated by the orthogonal decomposition of 
	\begin{align*}
		L^2(\R^3,\C^6) = \Jphys \orthsum \mathbf{G} 
		:= \ker \Div \orthsum \ran \Grad 
		= \ran \Rot \orthsum \ker \Rot 
	\end{align*}
	into transversal and longitudinal vector fields provided by the Helmholtz-Hodge-Weyl-Leray theorem (\cf Section~\ref{appendix:curl:Leray_decomposition}); it extends the unique splitting 
	\begin{align*}
		\Psi = \Rot \, \Phi + \Grad \, \varphi 
		, 
		&&
		\Phi \in L^2(\R^3,\C^6) 
		, \; 
		\varphi \in L^2_{\mathrm{loc}}(\R^3,\C^2) 
		, 
	\end{align*}
	from $\Cont^{\infty}_{\mathrm{c}}(\R^3,\C^6)$ to all of $L^2(\R^3,\C^6)$. Note that the vectors $\Rot \, \Phi$ and $\Grad \, \varphi$ are orthogonal with respect to the scalar product $\scpro{\cdot \,}{\cdot}_{L^2(\R^3,\C^6)}$, and thus there exist orthogonal projections $\mathbf{P}$ and $\mathbf{Q}$ onto $\Jphys$ and $\mathbf{G}$. Moreover, Remark~\ref{appendix:curl:remark:nomenclature} implies $\Cont^{\infty}_{\mathrm{c}}(\R^3,\C^6)$ and $H^1(\R^3,\C^6)$ are cores of $\Rot$. 
	
	The free Maxwell operator $\Rot \cong \int_{\BZ}^{\oplus} \dd k \, \Rot(k)$ is periodic with respect to any lattice, and thus we can use the Zak transform to fiber decompose it. The eigenvectors to any eigenvalue of $\Rot(k)$ can be explicitly constructed in terms of plane waves. 
	\begin{lem}[Band spectrum of $\Rot^{\Zak}$]\label{appendix:Maxwell:lem:properties_fibration_Rot}
		~
		\begin{enumerate}[(i)]
			\item $\displaystyle \sigma \bigl ( \Rot(k) \bigr ) = \{ 0 \} \cup \bigcup_{\gamma^* \in \Gamma^*} \bigl \{ \pm \sabs{\gamma^* + k} \bigr \}$
			\item There exists a $k$-dependent family of linearly independent vectors 
			\begin{align*}
				\bigl \{ u_{j \, \pm \, \gamma^*}(k) \; \; \vert \; \; \gamma^* \in \Gamma^* , \; j = 1 , 2 , 3 \bigr \} 
			\end{align*}
			which spans all of $L^2(\T^3,\C^6)$ and has the following properties: 
			\begin{enumerate}[(1)]
				\item The $u_{j \, \pm \, \gamma^*}(k)$ are eigenfunctions to $\Rot(k)$ with eigenvalues $\pm \sabs{\gamma^* + k}$ or $0$ for all $k \in \R^3$. 
				\item Away from $\Gamma^* \subset \R^3$, all maps $k \mapsto u_{j \, \pm \, \gamma^*}(k) \in L^2(\T^3,\C^6)$ are locally analytic on a small neighborhood which can be chosen to be independent of $j$ and $\gamma^*$. 
				\item Near $\gamma_0^* \in \Gamma^*$, only those $u_{j \, \pm \, \gamma^*}(k)$ are locally analytic on a common neighborhood for which $\gamma^* \neq - \gamma_0^*$ holds. 
			\end{enumerate}
		\end{enumerate}
	\end{lem}
	\begin{proof}
		We begin by analyzing the original operator $\Rot = \curl \otimes \sigma_2$ which can be factorized into an operator acting on $L^2(\R^3,\C^3)$ and a $2 \times 2$ matrix. The Pauli matrix $\sigma_2$ has eigenvalues $\pm 1$ and eigenvectors $w_{\pm}$. $\curl$ fibers in $\xi$ after applying the usual Fourier transform $\Fourier : L^2(\R^3,\C^3) \longrightarrow L^2(\R^3,\C^3)$, 
		\begin{align*}
			\Fourier \, \nabla_x^{\times} \, \Fourier^{-1} &= \int_{\R^3}^{\oplus} \dd \xi \, (\ii \xi)^{\times} 
			=: \int_{\R^3}^{\oplus} \dd \xi \, \curl(\xi)
			, 
		\end{align*}
		and $\curl(\xi) = \ii \xi^{\times}$ (see equation~\eqref{Maxwell:eqn:Maxwell_physical}) can be diagonalized explicitly: it has eigenvalues $\{ 0 , \pm \abs{\xi} \}$. Moreover, it can be seen that the eigenvectors $v_j(\xi)$, $j = 1 , 2 , 3$, are analytic away from $\xi = 0$. For $\xi \neq 0$, we set $v_1(\xi)$, $v_2(\xi)$ and $v_3(\xi)$ to be the eigenvectors to $+ \sabs{\xi}$, $- \sabs{\xi}$ and $0$, respectively. At $\xi = 0$ neither the eigenvalues $\pm \abs{\xi}$ nor the eigenvectors are analytic. 
		
		Now to the proof of the Lemma: For $j = 1 , 2 , 3$ let us set 
		\begin{align*}
			u_{j \, \pm \, \gamma^*}(k) := \e^{+ \ii \gamma^* \cdot y} \, v_j(\gamma^* + k) \otimes w_{\pm} 
		\end{align*}
		where $v_j(\gamma^* + k)$ is defined as in the preceding paragraph for $\xi = \gamma^* + k$. The exponential functions $\{ \e^{+ \ii \gamma^* \cdot y} \}_{\gamma^* \in \Gamma^*}$ and the $\{ v_j(\xi) \otimes w_{\pm} \}_{j = 1 , 2 , 3}$ form a basis of $L^2(\T^3)$ and $\C^3 \otimes \C^2 \cong \C^6$, respectively, and hence, the set of all $u_{j \, \pm \, \gamma^*}$ forms a basis of $L^2(\T^3,\C^6)$. Moreover, these vectors are eigenfunctions to $\Rot(k)$ with eigenvalues $\pm \sabs{\gamma^* + k}$ ($j = 1 , 2$) or $0$ ($j = 3$), and thus we have shown (i), $\sigma \bigl ( \Rot(k) \bigr ) = \{ 0 \} \cup \bigcup_{\gamma^* \in \Gamma^*} \bigl \{ \pm \sabs{\gamma^* + k} \bigr \}$, and (ii)~(1). 
		
		If $k_0 \in \R^3 \setminus \Gamma^*$, then 
		\begin{align*}
			\sabs{\gamma^* + k} \geq \mathrm{dist} \bigl ( k_0 , \Gamma^* \bigr ) > 0 
		\end{align*}
		is bounded from below which implies the eigenvectors $u_{j \, \pm \, \gamma^*}$ are analytic in some neighborhood of $k_0$. These vectors $v_j(\gamma^* + k)$, $j = 1 , 2 , 3$, are analytic on an open ball around $k_0$ with radius $\mathrm{dist} \bigl ( k_0 , \Gamma^* \bigr )$, proving (ii)~(2). 
		
		If, on the other hand, $k_0 = \gamma_0^* \in \Gamma^*$, then the basis involves the vector 
		\begin{align*}
			u_{j \, \pm \, -\gamma_0^*}(\gamma_0^*) = \e^{- \ii \gamma_0^* \cdot y} \, v_j(0) \otimes w_{\pm}
		\end{align*}
		which cannot be extended analytically to a neighborhood of $k_0 = \gamma_0^*$, thus proving (ii)~(3). 
	\end{proof}
	%
\end{appendix}

\printbibliography

\end{document}